\documentclass[letterpaper, 12pt, draftcls, onecolumn]{ieeeconf}

\IEEEoverridecommandlockouts                              

\usepackage[dvipdfmx]{graphicx}
\usepackage{amsmath, amssymb, bm, mathrsfs}
\usepackage{cite, color, graphicx, subcaption}

\newcommand{\cl}{\mathop{\rm Cl}\nolimits}

\newcommand{\argmin}{\mathop{\rm argmin}}

\newtheorem{theorem}{Theorem}[section]
\newtheorem{proposition}[theorem]{Proposition}

\newtheorem{assumption}[theorem]{Assumption}

\newtheorem{lemma}[theorem]{Lemma}
\newtheorem{remark}[theorem]{Remark}
\newtheorem{corollary}[theorem]{Corollary}

\title{\LARGE \bf
Stabilization of Discrete-time Piecewise Affine Systems 
with Quantized Signals 
}

\author{Masashi Wakaiki and Yutaka Yamamoto
\thanks{
This material is based upon work supported by 
The Kyoto University Foundation
and Murata Overseas Scholarship Foundation.
}
\thanks{
M. Wakaiki is with the Center for Control,
Dynamical-systems and Computation (CCDC), University of California,
Santa Barbara, CA 93106-9560 USA
(e-mail:{\tt  \ masashiwakaiki@ece.ucsb.edu}).
Y. Yamamoto is with the Department of Applied Analysis and Complex
Dynamical Systems, Graduate School of Informatics, Kyoto University, Kyoto
606-8501, Japan
(e-mail: {\tt \ yy@i.kyoto-u.ac.jp}).
}%
}

\begin{document}

\maketitle
\thispagestyle{empty}
\pagestyle{empty}

\begin{abstract}
This paper studies quantized control for 
discrete-time piecewise affine systems.
For given stabilizing feedback controllers, 
we propose an encoding strategy for local stability.
If the quantized state is near the boundaries of quantization regions, then
the controller can recompute a better quantization value.
For the design of quantized feedback controllers, 
we also consider the stabilization 
of piecewise affine systems with
bounded disturbances.
In order to derive a less conservative design method
with low computational cost,
we investigate a region to which 
the state belong in the next step.
\end{abstract}

\section{Introduction}
In many applications, 
the input and output of the controller
are quantized signals.
This is due to the physical properties of the actuators/sensors and
the data-rate limitation of links connected to the controller.
Quantized control for linear time-invariant systems actively studied from
various point of view, as surveyed in \cite{Nair2007, Ishii2012}.

Moreover, 
in the context of systems with 
discrete jumps such as switched systems and 
PieceWise Affine (PWA) systems,
control problems with limited information have recently
received increasing attention.
For sampled-data switched systems, 
a stability analysis under 
finite-level static quantization
has been developed in \cite{Wakaiki2014IFAC}, and 
an encoding and control strategy
for stabilization
has been proposed in the state feedback case \cite{Liberzon2014}, whose
related works have been presented 
for the output feedback case \cite{Wakaiki2014CDC}
and 
for the case with bounded disturbances \cite{Yang2015ACC}.
Also, our previous work \cite{WakaikiMTNS2014} has studied the stabilization
of continuous-time switched systems with quantized output feedback, 
based on 
the results in \cite{Brockett2000, Liberzon2003Automatica}.
However, relatively little work has been conducted on quantized control
for PWA systems.
In \cite{Xiaowu2013}, 
a sufficient condition for input-to-state stability
has been obtained for time-delay PWA systems with quantization signals,
but logarithmic quantizers in \cite{Xiaowu2013}
have an infinite number of quantization levels.

The main objective of this paper is to stabilize discrete-time 
PWA systems with quantized signals.
In order to achieve the local asymptotic 
stabilization of discrete-time PWA plants
with finite data rates,
we extend the event-based encoding method in 
\cite{Brockett2000, Liberzon2000workshop}.
It is assumed that we are given feedback controllers that stabilize 
the closed-loop system in the sense that there exists a
piecewise quadratic Lyapunov function.
In the input quantization case, 
the controller receives the original state. On the other hand,
in the state quantization case,
the quantized state and the currently active mode of
the plant are available to the controller. 
The information on the active mode prevents a mode mismatch
between the plant and the controller, and moreover,
allows the controller side to recompute 
a better quantization value if the quantized state transmitted from
the quantizer is near the boundaries of quantization regions.
This recomputation is motivated in Section 7.2 in \cite{Liberzon2014}.

We also investigate the design of 
quantized feedback controllers.
To this end,
we consider the stabilization problem of discrete-time 
PWA systems with bounded disturbances (under no quantization).
The Lypunov-based stability analysis and stabilization
of discrete-time PWA systems has been studied in 
\cite{Feng2002, Ferrari2002} and 
\cite{Cuzzola2002, Lazar2008, Xu2014}
in terms of Linear Matrix Inequalities (LMIs) and Bilinear Matrix Inequalities (BMIs).
In proofs that Lyapunov functions decrease along the trajectories of 
PWA systems, 
the one-step reachable set, that is,
the set to which the state belong in one step, plays an important role.
In stability analysis, the one-step reachable set
can be obtained by linear programming.
By contrast,  in the stabilization case, 
since the next-step state depends on the control input,
it is generally difficult to obtain the one-step reachable set.
Therefore
many previous works for the design of stabilizing controllers
assume that the one-step reachable set is the total state space.
However, 
if disturbances are bounded, then
this assumption leads to conservative results and
high computational loads as the number of the plant mode increases.

We aim to find the one-step reachable set for PWA systems with bounded
disturbances. To this effect, we derive
a sufficient condition on feedback controllers for 
the state to belong to a given polyhedron in one step.
This condition can be used to add constraints on the state and
the input as well.
Furthermore, we obtain a set containing the one-step reachable set
by using the information of the input matrix $B_i$ 
and the input bound $u \in \mathbf{U}$.
This set is conservative because 
the affine feedback structure $u = K_i x + g_i$ for mode $i$
is not considered, but it can be used 
when we design the polyhedra that are assumed to be
given in the above sufficient condition.
Combining the proposed condition with results in 
\cite{Cuzzola2002, Lazar2008, Xu2014} for
Lyapunov functions to be positive and decrease along the trajectories,
we can design stabilizing controllers for PWA systems with bounded
disturbances.

This paper is organized as follows.
The next section shows 
a class of quantizer and a basic assumption on stability.
In Sections III and IV, we present
an encoding strategy 
to achieve local stability for PWA systems
in the input quantization case and the state quantization case,
respectively.
In Section V, we study the one-step reachable set for
the stabilization problem of 
PWA systems with bounded disturbances.
Finally, concluding remarks are given in Section VI.

Due to space constraints, 
all proofs and a numerical example 
have been omitted and 
can be found in \cite{}.

{\it Notation}:
For a set $E \subset \mathbb{R}^{{\sf n} }$,
we denote by $\cl(E)$ the closure of $E$.
For sets $E_1,E_2 \subset \mathbb{R}^{{\sf n}}$,
let $E_1 \oplus E_2 =
\{
v + u:~ v \in E_1,~u \in E_2 
\}$
denote their Minkowski sum.

Let $\lambda_{\min}(P)$ and $\lambda_{\max}(P)$ denote
the smallest and the largest eigenvalue of 
$P \in \mathbb{R}^{{\sf n}\times {\sf n}}$.
Let $M^{\top}$ denote the transpose of
$M \in \mathbb{R}^{\sf m\times n}$.
For $v \in \mathbb{R}^{\sf n}$,
we denote 
the $l$-th entry of $v$ by $v^{(l)}$.
Let $\bf 1$ be a vector all of whose entries are one.
For vectors $v,u \in \mathbb{R}^{\sf n}$, the inequality 
$v \leq u$ means that
$v^{(l)} \leq u^{(l)}$ for every $l=1,\dots, {\sf n}$.
On the other hand, for a square matrix $P$,
the notation $P \succeq 0$ ($P \succ 0$) means that $P $ is symmetric and
semi-positive (positive) definite.

The Euclidean norm of $v \in \mathbb{R}^{\sf n}$ is denoted by
$|v| = (v^*v)^{1/2}$. The Euclidean induced norm of
$M \in \mathbb{R}^{\sf m\times n}$ is defined by
$\|M\| = \sup \{  |Mv |:~v\in \mathbb{R}^{\sf n},~|v|= 1 \}$.
The $\infty$-norm of 
$v=[v_1~\!\dotsb ~\!v_{\sf n} ]^{\top}$ 
is denoted by $|v|_{\infty} = \max\{|v_1|,\dots, |v_{\sf n}|\}$, and
the induced norm of $M \in \mathbb{R}^{\sf m\times n}$ corresponding
to the $\infty$-norm is defined by
$\|M\|_{\infty} = \sup \{  |Mv |_{\infty}:~
v\in \mathbb{R}^{\sf n},~|v|_{\infty}= 1 \}$.
For $r > 0$, let $\mathbf{B}_r = 
\{x \in \mathbb{R}^{\sf n}:~ |x| \leq r\}$
and 
$\mathbf{B}_r^{\infty} = 
\{x \in \mathbb{R}^{\sf n}:~ |x|_{\infty} \leq r\}$.

\section{Quantized Control of PWA systems}
We consider the following class of discrete-time PWA systems:
\begin{equation}
\label{eq:PWAS}
x_{k+1} = A_ix_k + B_i u_k + f_i =: G_i(x_k,u_k)\quad 
(x_k  \in \mathcal{X}_i),
\end{equation}
where $x_k \in \mathbf{X} \subseteq \mathbb{R}^{\sf{n}}$ 
is the state and $u_k \in \mathbb{R}^{\sf{m}}$ is the control input.
The set $\mathbf{X}$ is divided into finitely many disjoint 
polyhedra\footnote{
	A polyhedron is the intersection of finitely many halfspaces.
}
$\mathcal{X}_1,\dots,\mathcal{X}_s$: 
${\bf X} = \sum_{i=1}^{s}\mathcal{X}_i $.
We denote the index set $\{1,2,\dots,s\}$ by $\mathcal{S}$. 

Given a feedback gain $K_i \in \mathbb{R}^{\sf{n} \times \sf{m}}$ 
and an affine term $g_i \in \mathbb{R}^{\sf m}$ for 
each mode $i=1,\dots,s$, 
the control input is in the affine state feedback form:
\begin{equation}
\label{eq:input}
u_{k} = K_i x_k + g_i \qquad (x_k  \in \mathcal{X}_i).
\end{equation}
We assume that $f_i = g_i = 0$ if $0 \in \cl(\mathcal{X}_i)$.
We will study 
the design of $K_i$ and $g_i$ in Section V,
but for quantized control in Sections III and IV, 
$K_i$ and $g_i$ are assumed to be given.


\subsection{Quantizers}
In this paper, we use the class of quantizers proposed 
in \cite{Liberzon2003Automatica}.

Let $\mathcal{P}$ be a set composed of finitely many points in 
$\mathbb{R}^{\sf{N}}$.
A quantizer $q$ is a piecewise constant function from
$\mathbb{R}^{\sf N}$ to $\mathcal{P}$.
Geometrically, this means that
$\mathbb{R^{\sf{N}}}$ is divided into
a finite number of quantization regions of the form
$\{\xi \in \mathbb{R^{\sf{N}}}:~q(\xi) = q_p \}$ $(q_p \in \mathcal{P})$.
For the quantizer $q$, we assume that
there exist $M, \Delta$ with $M > \Delta > 0$ such that
\begin{equation}
\label{eq:quantizer_cond1}
|\xi| \leq M \quad \Rightarrow \quad |q(\xi) - \xi|_{\infty} \leq \Delta .
\end{equation}

The condition \eqref{eq:quantizer_cond1} gives an upper bound on
the quantization error if the quantizer saturates.
In this paper, we assume that 
a bound on the magnitude of the initial state is known, 
and hence we do not use
a condition in the case when the quantizer saturates.

We use quantizers with an adjustable parameter $\mu>0$:
\begin{equation}
\label{eq:zoom_q}
q_{\mu}(\xi) = \mu q\left( \frac{\xi }{\mu}\right).
\end{equation}
The quantized value $q_{\mu_k}(\xi_k)$ is the data on $\xi_k$
transmitted to the controller at time $k$. We adjust
$\mu_k$ to 
obtain detailed information on $\xi_k$ near the origin.

\subsection{Assumption on stability}
Define 
\begin{equation}
\label{eq:reachable_set_def}
\mathcal{R}_i := \{ G_i(x,K_ix+g_i):~x\in \mathcal{X}_i\}, 
\end{equation}
which
is the one-step reachable set from $\mathcal{X}_i$ for the 
PWA system \eqref{eq:PWAS} and
the state feedback law \eqref{eq:input}
without quantization.
Define also
\begin{equation}
\label{eq:calB_def}
\mathcal{B}_i :=
\begin{cases}
\{B_id:~|d|_{\infty} \leq \Delta\} 
& \text{(input quantization case)} \\
\{B_iK_id:~|d|_{\infty} \leq \Delta\} 
& \text{(state quantization case)}		
\end{cases}
\end{equation}
We assume that the following stability of the closed-loop system 
is guaranteed by
a piecewise quadratic Lyapunov function:
\begin{assumption}
	\label{ass:Lyapunov_for_QS}
	{\em 
		Consider the PWA system \eqref{eq:PWAS} with
		given affine feedback \eqref{eq:input}.
		Define a function $V_i: \mathcal{X}_i \to \mathbb{R}$ by 
		\begin{equation}
		\label{eq:Lypunov_func_Feng}
		V_i(x) := 
		\begin{cases}
		x^{\top} P_i x & 0 \in \cl (\mathcal{X}_i) \\
		\begin{bmatrix}
		x \\ 1 
		\end{bmatrix}^{\top}
		{\bar P}_i
		\begin{bmatrix}
		x \\ 1 
		\end{bmatrix}
		&0 \not\in \cl (\mathcal{X}_i),
		\end{cases}
		\end{equation}
		where $P_i \in \mathbb{R}^{{\sf n} \times {\sf n}}$ and
		$\bar{P}_i \in \mathbb{R}^{({\sf n}+1) \times ({\sf n}+1)}$
		are symmetric matrices.
		There exist $\alpha, \beta > 0$ and $\gamma_i > 0$ for 
		$i\in \mathcal{S}$, such that 
		the Lypunov function $V:\mathbf{X} \to \mathbb{R}$ 
		defined by $V(x) := V_i(x)$ $(x \in \mathcal{X}_i)$
		satisfies 
		\begin{gather}
		\alpha |x|^2 \leq V(x) \leq \beta |x|^2 
		\label{eq:Lyapunov_bound} \\
		V_j((A_i+B_iK_i)x+f_i+B_ig_i) - V_i(x) \leq -\gamma_i |x|^2
		\label{eq:diff_VjVi}
		\end{gather}
		for every $i \in \mathcal{S}$, 
		$j \in \mathcal{S}_i$, and $x \in \mathcal{X}_i$, where
		$\mathcal{S}_i$ is defined by
		\begin{equation}
		\label{eq:Sei_def}
		\mathcal{S}_i :=
		\left\{
		j \in \mathcal{S}:~
		\mathcal{X}_j  \cap
		\bigl(\mathcal{R}_i \oplus \mathcal{B}_i \bigr)
		\not = \emptyset
		\right\}.
		\end{equation}
	}
\end{assumption}
In Section V, 
we will discuss how to obtain $\mathcal{S}_i$ of \eqref{eq:Sei_def}
in the design process of $K_i$ and $g_i$.

\section{Input Quantization Case}
In this section, we study stabilization with quantized input:
\begin{equation*}
u_k = q(K_i x_k + g_i) \qquad (x_k \in \mathcal{X}_i).
\end{equation*}
The closed-loop system we consider is given by
\begin{align}
x_{k+1} 
&= A_ix_k + B_i q(K_i x_k+g_i) + f_i \qquad (x_k  \in \mathcal{X}_i) \notag \\
&= 
G_i(x_k,K_ix_k+g_i)+ B_i (q(K_ix_k+g_i) - (K_i x_k + g_i)).
\label{eq:input_saturation_stateEq}
\end{align}
We place the following assumption on the state transition:
\begin{assumption}
	\label{ass:input_ver}
	{\it
		Define 
		$\mathcal{B}_i := 
		\{
		B_id: |d|_{\infty} \leq \Delta
		\}$.
		For every $i \in \mathcal{S}$, the 
		one-step reachable set $\mathcal{R}_i$ in \eqref{eq:reachable_set_def}
		satisfies
		$
		\mathcal{R}_i \oplus \mathcal{B}_{i} 
		\subset {\bf X}$.
	}
\end{assumption}

The condition
$
\mathcal{R}_i \oplus \mathcal{B}_{i} 
\subset {\bf X}$
implies 
that $\bf X$ is invariant for the system \eqref{eq:input_saturation_stateEq},
and
checking this condition 
is closely related to how to derive $\mathcal{S}_i$ in \eqref{eq:Sei_def}.
In Section V, we will derive sufficient conditions on $K_i$
and $g_i$ for $\mathcal{R}_i \oplus \mathcal{B}_{i} 
\subset \mathbf{X}$ to hold; 
see Remark \ref{rem:Total_ss_cond}.

First we fix the zoom parameter $\mu=1$.
Similarly to \cite{Bullo2006, Liberzon2003Automatica, Brockett2000},
we show that the Lyapunov function decreases until the state gets to
the corresponding level set.
\begin{theorem}
	\label{thm:MKtom}
	{\em
		Consider the PWA system\eqref{eq:input_saturation_stateEq}
		with given $K_i$ and $g_i$.
		Let Assumptions \ref{ass:Lyapunov_for_QS} and
		\ref{ass:input_ver} hold.
		Fix $\varepsilon_{ij}, \delta_{ij} \in (0,1)$, and define
		\begin{align}
		Q_j &\!:=\!
		\begin{cases}
		P_j & 0 \in \cl (\mathcal{X}_j) \\
		\begin{bmatrix}
		I_{{\sf n} \times {\sf n}} & 0_{{\sf n} \times 1}
		\end{bmatrix}
		\bar P_j
		\begin{bmatrix}
		I_{{\sf n} \times {\sf n}} \\ 0_{{\sf n} \times 1}
		\end{bmatrix}
		& 0 \not\in \cl (\mathcal{X}_j)
		\end{cases} \label{eq:Qj_def} \\
		h_{ij} &\!:=\!
		\begin{cases}
		P_j(f_i+B_ig_i) & 0 \in \cl (\mathcal{X}_j) \\
		\begin{bmatrix}
		I_{{\sf n} \times {\sf n}} & 0_{{\sf n} \times 1}
		\end{bmatrix}
		\bar P_j
		\begin{bmatrix}
		f_i+B_ig_i \\ 1
		\end{bmatrix}
		& 0 \not\in \cl (\mathcal{X}_j)
		\end{cases} \label{eq:gij_def} \\
		\phi_{1,ij} &\!:=\! 
		{\sf m}
		\left(
		\frac{\|B_i^{\top} Q_j B_i\| }{(1\! - \! \varepsilon_{ij}) \delta_{ij} \gamma_{i}}
		+
		\frac{\|(A_i \! +\! B_iK_i)^{\top}
			Q_j B_i\|^2}{((1\! - \! \varepsilon_{ij}) \gamma_i)^2 \delta_{ij} (1-\delta_{ij})}
		\right)
		\notag \\
		\phi_{2,ij} &\!:=\! 
		\frac{2\sqrt{\sf m} \|h_{ij}^{\top}B_i \|}{(1-\varepsilon_{ij}) \delta_{ij} \gamma_{i}}. \notag 
		\end{align}
		Also, let $M > |g_i|$ for all $i \in \mathcal{S}$, and set
		\begin{gather*}
		m_{i} \!:=\! 
		\max_{j \in \mathcal{S}_i}
		\sqrt{\phi_{1,ij}\Delta^2 + \phi_{2,ij}\Delta},\quad
		m \!:=\! \max_{i\in \mathcal{S}} m_{i}
		\\
		\varepsilon_i \!:=\! \max_{j \in \mathcal{S}^e_i} \varepsilon_{i,j}, \quad
		M_K \!:=\! \min_{i \in \mathcal{S}} \frac{M - |g_i|}{\|K_i\|}.
		\end{gather*}
		Define $\mathcal{E}_{M_K}$ and $\mathcal{E}_m$ by
		\begin{align*}
		\mathcal{E}_{M_K} := \{ x : V(x) \leq \alpha M_K^2\}, \quad
		\mathcal{E}_m := \{ x : V(x) \leq \beta m^2\}.
		\end{align*}
		If $m$ and $M_K$ satisfy
		\begin{equation}
		\label{eq:Km_M}
		\beta m^2 < \alpha M_K^2,
		\end{equation}
		then all solutions
		of \eqref{eq:input_saturation_stateEq}
		that start in $\mathcal{E}_{M_K} \cap \mathbf{X}$ 
		enter $\mathcal{E}_m$
		in a finite time $k_0$ satisfying
		\begin{equation}
		\label{eq:k0_cond_input}
		0 \leq
		k_0 \leq \frac{\alpha M_K^2 - \beta m^2 }
		{\min_{i \in \mathcal{S}} \left( \varepsilon_i \gamma_i m_i^2 \right)}
		=: \bar{k}_0.
		\end{equation}
		Furthermore, if 
		\begin{equation}
		\label{eq:invariant_cond}
		\max_{i \in \mathcal{S}} \left( \|A_i\| m + \| B_i \| (\|K_i\| m+
		\sqrt{\sf m}\Delta) + |f_i| \right) \leq 
		\sqrt{\frac{\alpha}{\beta}} M_K
		\end{equation}
		holds, then
		the solution $x_k$ belongs to 
		$\mathcal{E}_{M_K} \cap \mathbf{X}$
		for all $k \geq 0$.
	}
\end{theorem}

\begin{proof}
	In order to utilize \eqref{eq:diff_VjVi},
	first we show that if $x_k \in \mathcal{X}_i \cap \mathcal{E}_{M_K}$,
	then there exists $j \in \mathcal{S}_{i}$ such that $x_{k+1} \in \mathcal{X}_j$.
	Suppose that $x_k \in \mathcal{X}_i \cap \mathcal{E}_{M_K}$.
	Define $d_k$ by $d_k := q(K_ix_k+g_i)-(K_ix_k+g_i)$. Then
	we have $x_{k+1} = G_i(x_k,K_ix_k+g_i) + B_i d_k$.
	Since $x_k \in \mathcal{X}_i$, it follows that
	$G_i(x_k,K_ix_k + g_i) \in \mathcal{R}_i$.
	Moreover, $x_k \in \mathcal{E}_{M_K}$ implies $|K_ix_k+g_i| \leq M$, and hence
	$|d_k|_{\infty} \leq \Delta $ from \eqref{eq:quantizer_cond1} and
	$B_i d_k \in \mathcal{B}_i$. We therefore obtain
	\begin{equation}
	\label{eq:x_k+1_RB}
	x_{k+1} \in \mathcal{R}_i \oplus \mathcal{B}_i.
	\end{equation}
	Therefore, 
	from Assumption \ref{ass:input_ver},
	there exists $j \in \mathcal{S}$ such that 
	\begin{equation}
	\label{eq:x_k+1_Xj}
	x_{k+1} \in \mathcal{X}_j.
	\end{equation}
	Combining \eqref{eq:x_k+1_RB} and \eqref{eq:x_k+1_Xj}, 
	we have 
	$\mathcal{X}_j  \cap (\mathcal{R}_i \oplus \mathcal{B}_i) 
	\not= \emptyset$. Thus
	$j \in \mathcal{S}_i$ by definition.
	
	In what follows, for simplicity of notation, we omit
	the indices $i$ and $j$ of $\varepsilon_{ij}$, $\delta_{ij}$, $\gamma_{i}$, 
	$\phi_{1,ij}$, and $\phi_{2,ij}$. 
	Define $\bar{A}_i := A_i+B_iK_i$ and $e_k := B_id_k$.
	Since \eqref{eq:diff_VjVi} holds, we have 
	\begin{align*}
	V(x_{k+1}) - V(x_k) 
	&\leq
	-\gamma |x_k|^2 + 2|\bar A_i^{\top}Q_j e_k| \cdot |x_k| +
	e_k^{\top}Q_je_k + 2 h_{ij} e_k \\
	&= 
	-\varepsilon \gamma |x_k|^2 - (1-\varepsilon) (1-\delta) \gamma |x_k|^2
	-  (1-\varepsilon) \delta \gamma |x_k|^2 \\[-2pt]
	&\qquad \quad+ 2|\bar A_i^{\top}Q_j e_k| \cdot |x_k| +
	e_k^{\top}Q_je_k + 2 h_{ij} ^{\top} e_k \\
	&=
	-\varepsilon \gamma |x_k|^2
	- (1\!-\! \varepsilon) (1\!-\!\delta) \gamma 
	\left( 
	|x_k| \!-\! \frac{|\bar A_i^{\top}Q_j e_k|}{(1\!-\!\varepsilon) (1\!-\!\delta) \gamma } 
	\right)^2 \\[-2pt]
	&\qquad \quad - (1\!-\! \varepsilon) \delta \gamma |x_k|^2 
	+
	e_k^{\top}Q_je_k  
	+ \frac{|\bar A_i^{\top}Q_j e_k|^2 }{(1\!-\! \varepsilon) (1\!-\! \delta) \gamma }
	+ 2h_{ij}^{\top}e_k \\
	&\leq
	-\varepsilon \gamma |x_k|^2 - \Upsilon,
	\end{align*}
	where 
	\begin{equation*}
	\Upsilon := (1\!-\!\varepsilon) \delta \gamma
	|x_k|^2 \!-\! 
	e_k^{\top}Q_je_k \!-\! \frac{|\bar A_i^{\top}Q_j e_k|^2 }
	{(1\!-\!\varepsilon) (1\!-\!\delta) \gamma }
	\!-\!
	2h_{ij} ^{\top}e_k.
	\end{equation*}
	If $x_k \in \mathcal{E}_{M_K}$, then
	\[
	|d_k| 
	\leq
	\sqrt{\sf m} |d_k|_{\infty} =
	\sqrt{\sf m}
	|K_ix_k+g_i - q(K_i x_k+g_i)|_{\infty}
	\leq \sqrt{\sf m} \Delta.
	\]
	Hence, noticing $e_k = B_id_k$, we have
	\begin{align*}
	e_k^{\top}Q_je_k  + \frac{|\bar A_i^{\top}Q_j e_k|^2 }
	{(1-\varepsilon) (1-\delta) \gamma } 
	&\leq
	\left(
	\|B_i^{\top} Q_j B_i\| +
	\frac{\|\bar A_i^{\top}Q_j B_i\|^2 }{(1-\varepsilon) (1-\delta) \gamma }
	\right) \cdot {\sf m}\Delta^2 \\
	h_{ij} ^{\top} e_k 
	&\leq \|h_{ij}^{\top} B_i \| \cdot \sqrt{\sf m} \Delta.
	\end{align*}
	Therefore 
	\begin{align*}
	\frac{\Upsilon}{(1-\varepsilon) \delta \gamma} 
	\geq
	|x_k|^2 - \phi_1 \Delta^2 - \phi_2 \Delta  
	\geq
	|x_k|^2 - m_{i}^2.
	\end{align*}
	For every $i \in \mathcal{S}$, we obtain
	\begin{equation}
	\label{eq:Lyapunov_decrease}
	V(x_{k+1}) - V(x_{k}) \leq -\varepsilon_i \gamma_i |x_k|^2
	\leq - \min_{i \in \mathcal{S}} \left( \varepsilon_i \gamma_i m_i^2\right)
	\end{equation}
	whenever $|x_k| \geq m_i$.
	Note that the most right side of \eqref{eq:Lyapunov_decrease} 
	is independent of the plant mode $i$.
	
	By \eqref{eq:Km_M}, we have
	\[
	\mathbf{B}_m \subset
	\mathcal{E}_{m} \subset
	\mathcal{E}_{M_K}.
	\]
	In conjunction with \eqref{eq:Lyapunov_decrease}, this shows that
	if the initial state $x_0$ belongs to $\mathcal{E}_{M_K}$,
	then $x_{k_0} \in \mathcal{E}_{m}$ holds for some integer $k_0$
	satisfying \eqref{eq:k0_cond_input}.
	
	Let us next prove that $\mathcal{E}_{M_K}\cap \mathbf{X}$ is 
	an invariant region for the system \eqref{eq:input_saturation_stateEq}.
	From \eqref{eq:Lyapunov_decrease},
	$x_{k} \in \mathcal{E}_{M_K}$ until $x_k \not\in \mathcal{B}_m$.
	Once $x_k \in \mathcal{B}_m$,
	we have
	\begin{equation*}
	|x_{k+1}| \leq \| A_i\| m + \| B_i \| (\|K_i\| m+ \sqrt{\sf m} \Delta) 
	+ |f_i|.
	\end{equation*}
	Therefore if 
	\eqref{eq:invariant_cond} holds, 
	then $x_{k} \in \mathbf{B}_m$ leads to $x_{k+1} \in \mathcal{E}_{M_K}$.
	The state trajectories again go to $\mathbf{B}_m$
	while belonging to $\mathcal{E}_{M_K}$
	Since Assumption \ref{ass:input_ver} gives 
	$x_k \in \mathbf{X}$ for all $k \geq 0$, 
	we see that $x_k \in \mathcal{E}_{M_K}\cap \mathbf{X}$ for all
	$k \geq 0$.
\end{proof}

As in \cite{Brockett2000},
we can achieve the state convergence to the origin
by adjusting the zoom parameter $\mu$:
\begin{theorem}
	\label{thm:Input_convergence}
	{\it
		Consider the PWA system \eqref{eq:input_saturation_stateEq}
		with given $K_i$ and $g_i$.
		Let Assumptions \ref{ass:Lyapunov_for_QS} and 
		\ref{ass:input_ver} hold.
		Let the initial state $x_0 \in \mathcal{E}_{M_K} \cap {\bf X}$ and
		the initial zoom parameter $\mu_0 = 1$.
		Assume that \eqref{eq:Km_M} and \eqref{eq:invariant_cond} hold, 
		and define
		\[
		\Omega := \sqrt{\frac{\beta}{\alpha}} \cdot
		\frac{m}{M_K} < 1.
		\]
		Adjust $\mu$ by $\mu_{k} = \Omega\mu_{k-1}$
		when $x_k$ gets to $\mathbf{B}_{\mu_{k-1} m}$, and 
		send to the plant 
		the quantized input $q_{\mu_k}(K_ix_k+g_i)$ at time $k$ if 
		$x_k \in \mathcal{X}_i$.
		This event-based update strategy of $\mu$
		leads to $x_k \to 0~(k \to \infty)$.
	}
\end{theorem}
\begin{proof}
	First we
	prove that 
	as long as the quantizer does not saturate,
	the state trajectory belongs to $\mathbf{X}$ and 
	the \eqref{eq:diff_VjVi} holds for all $k \geq 0$.
	Define $\mathcal{B}_i^{[p]} :=
	\{ B_i d : |d|_{\infty} \leq \Omega^p \Delta
	\}$.
	Since $\mathcal{B}_i^{[p]} \subset \mathcal{B}_i$,
	if Assumption \ref{ass:input_ver} holds, then
	$
	\mathcal{R}_i \oplus \mathcal{B}_i^{[p]}
	\subset {\bf X}
	$ $(i \in \mathcal{S})$
	for every $p \geq 0$.
	Hence $x_k \in \mathbf{X}$ for all $k \geq 0$
	unless the quantizer saturates.
	Moreover, if we define
	$\mathcal{S}_i^{[p]}$ by the set
	consisting of all $j \in \mathcal{S}$ satisfying
	$\mathcal{X}_j \cap \bigl(\mathcal{R}_i \oplus 
	\mathcal{B}_{\Omega^p \Delta} \bigr) 
	\not = \emptyset$ as in \eqref{eq:Sei_def},
	then $\mathcal{S}_i^{[p]} \subset \mathcal{S}_i^{[0]} = \mathcal{S}_i$.
	Thus \eqref{eq:diff_VjVi} holds for every $i \in \mathcal{S}$ and
	$j \in \mathcal{S}_i^{[p]}$, and hence we have \eqref{eq:diff_VjVi}
	for all $k \geq 0$ unless the quantizer saturation occurs.
	
	Let an update occur at $k=\ell_0$, i.e.,
	$x_{\ell_0} \in \mathbf{B}_{\mu_{\ell_0-1} m}$
	and
	$\mu_{\ell_0} = \Omega \mu_{\ell_0 - 1}$. Then
	we have
	\[
	\beta( \mu_{\ell_0-1}m)^2
	= 
	\alpha( \mu_{\ell_0}M_K)^2.
	\]
	Therefore $\mathcal{E}_{\mu_{\ell_0} M_K}$ defined by
	\begin{align*}
	\mathcal{E}_{\mu_{\ell_0} M_K} := 
	\{ x : V(x) \leq \alpha ( \mu_{\ell_0} M_K)^2\}
	\end{align*}
	satisfies $\mathcal{E}_{\mu_{\ell_0} M_K} = 
	\mathcal{E}_{\mu_{\ell_0-1} m} := 
	\{ x : V(x) \leq \beta( \mu_{k-1}m)^2\}
	\supset \mathbf{B}_{\mu_{\ell_0-1} m}$.
	Since $x_{\ell_0} \in \mathbf{B}_{\mu_{\ell_0-1} m}$,
	it follows that
	$x_{\ell_0} \in \mathcal{E}_{\mu_{\ell_0} M_K}$.
	Hence Theorem \ref{thm:MKtom} shows that for all $k \geq 0$,
	$x_k \in \mathcal{E}_{\mu_{k} M_K}$, which means
	$|x_k| \leq \mu_k M_K$ and the quantizer does not 
	saturate for every $k \geq 0$.
	Moreover, the update period does not exceed $\bar{k}_0$ in \eqref{eq:k0_cond_input}.
	Since $\Omega < 1$, it follows that $\mu_k \to 0$ as $k \to \infty$.
	Thus $x_k \to 0$ as $k \to \infty$.
\end{proof}

\begin{remark}
	For {\em continuous-time} systems, the level sets $\mathcal{E}_{M_K}$
	and $\mathcal{E}_m$ are invariant regions of
	the state trajectories~\cite{Liberzon2003Automatica}.
	However, for {\em discrete-time} systems, 
	$\mathcal{E}_m$ may not be invariant. We therefore need the
	event-based adjustment of the zoom parameter as in
	\cite[Section III]{Brockett2000} and
	\cite{Liberzon2000workshop}.
\end{remark}

\section{State Quantization Case}
Let us next study stabilization of PWA systems with quantized state feedback.

We assume that 
the controller receives the information on the quantized state and the active mode.
\begin{assumption}
	\label{ass:Con_QS_AM}\
	{\it
		The quantizer has the information on 
		the switching regions $\{\mathcal{X}_i\}_{i\in\mathcal{S}}$.
		The quantizer sends to the controller the information on
		the quantized state and the active mode.
	}
\end{assumption}
Under Assumption \ref{ass:Con_QS_AM},
the control $u_k$ is given by
\begin{align*}
u_k = K_i q(x_k) + g_i \qquad (x_k \in \mathcal{X}_i).
\end{align*}
The closed-loop system we consider can be written in this way:
\begin{align}
x_{k+1} &= A_ix_k + B_i K_i q(x_k) + f_i +B_ig_i
\qquad
(x_k  \in \mathcal{X}_i)  \notag \\
&=
G_i(x_k,K_ix_k+g_i) + B_iK_i (q(x_k) - x_k).
\label{eq:state_saturation_stateEq}
\end{align}

\subsection{Stability analysis}
We place an assumption similar to Assumption \ref{ass:input_ver}. 
\begin{assumption}
	\label{ass:state_ver}
	{\it
		Define 
		$\mathcal{B}_i := 
		\{
		B_iK_id: |d|_{\infty} \leq \Delta
		\}$.
		For every $i \in \mathcal{S}$,
		the 
		one-step reachable set 
		$\mathcal{R}_i$ in \eqref{eq:reachable_set_def}
		satisfies
		$
		\mathcal{R}_i \oplus \mathcal{B}_{i} 
		\subset {\bf X}$.
	}
\end{assumption}

See Remark \ref{rem:Total_ss_cond} for the condition 
$
\mathcal{R}_i \oplus \mathcal{B}_{i} 
\subset {\bf X}
$.

As in the input quantization case, we first fix $\mu = 1$ and
obtain a result similar to Theorem \ref{thm:MKtom}, based on
the technique in \cite{Brockett2000}.

\begin{theorem}
	\label{thm:Mtom}
	{\it
		Consider the PWA system \eqref{eq:state_saturation_stateEq}
		with given $K_i$ and $g_i$.
		Let Assumptions \ref{ass:Lyapunov_for_QS} and 
		\ref{ass:state_ver} hold.
		Fix $\varepsilon_{ij}, \delta_{ij} \in (0,1)$.
		Define $Q_j$ and $h_{ij}$ as in \eqref{eq:Qj_def}
		and \eqref{eq:gij_def} respectively, and 
		define $\phi_{1,ij}$ and $\phi_{2,ij}$ by
		\begin{align*}
		\phi_{1,ij} &\!:=\! 
		{\sf n} \! 
		\left(
		\frac{
			\|K_i^{\top}B_i^{\top} Q_j B_iK_i\| 
		}{(1\!-\! \varepsilon_{ij}) \delta_{ij} \gamma_{i}}
		\!+ \! \frac{\|(A_i \!+ \! B_iK_i)^{\!\top}Q_j B_iK_i\|^2}
		{((1\!- \! \varepsilon_{ij}) \gamma_{i})^2 (1\!- \!\delta_{ij}) \delta_{ij} } 
		\right) \\
		\phi_{2,ij} &\!:=\!
		\frac{2\sqrt{\sf n}\|h_{ij}^{\top}B_i K_i\|}
		{(1-\varepsilon_{ij}) \delta_{ij} \gamma_{i}}.
		\end{align*}
		Set $m_{i,j}$,
		$m_i$, $m$, $\alpha_{\min}$, and $\beta_{\max}$
		as in Theorem \ref{thm:MKtom},
		and set 
		\begin{equation*}
		\tilde m := m + \sqrt{\sf n} \Delta,\qquad
		\bar m := m + 2\sqrt{\sf n} \Delta.
		\end{equation*}
		Define $\mathcal{E}_{M}$ and $\mathcal{E}_{\bar m}$ by
		\begin{align*}
		\mathcal{E}_{M} := \{ x : V(x) \leq \alpha M^2\}, \quad
		\mathcal{E}_{\bar m} := \{ x : V(x) \leq \beta \bar{m}^2\}.
		\end{align*}
		If $M$ satisfies
		\begin{equation}
		\label{eq:Km_M_state}
		\beta \bar{m}^2 < \alpha  M^2,
		\end{equation}
		then 
		all solutions
		of \eqref{eq:state_saturation_stateEq}
		that start in $\mathcal{E}_{M} \cap \mathbf{X}$ 
		enter $\mathcal{E}_{\bar m}$
		in a finite time $k_0$ satisfying
		\begin{equation}
		\label{eq:k0_cond}
		0 \leq
		k_0 \leq \frac{\alpha ( M^2 - m^2) }
		{\min_{i \in \mathcal{S}} \left( \varepsilon_i \gamma_i m_i^2 \right)}
		=: \bar k_0,
		\end{equation}
		and $x \in \mathcal{E}_{\bar m}$ can be observed
		from $q(x) \in \mathbf{B}_{\tilde m}$.
		Furthermore, 
		if 
		\begin{equation}
		\label{eq:invariant_cond_state_ver}
		\max_{i \in \mathcal{S}} \left( \|A_i\| \bar m + 
		\| B_i K_i\| \tilde m + |f_i+B_ig_i| \right) \leq 
		\sqrt{\frac{\alpha }{\beta }} M,
		\end{equation}
		then the solution belongs to $\mathcal{E}_{M} \cap \mathbf{X}$
		for all $k \geq 0$.
	}
\end{theorem}
\begin{proof}
	If we define $e_k := B_iK_i d_k$, then
	the proof follows the same lines as that of Theorem 3.7
	until \eqref{eq:Lyapunov_decrease}.
	We see that
	the Lyapunov function decreases
	if the initial state belongs to $\mathcal{E}_{M} \cap \mathbf{X}$
	and 
	if the state does not arrive at $\mathbf{B}_m$.
	
	We show that
	the quantized state $q(x_k)$ gets to $\mathbf{B}_{\tilde m}$ at 
	$k \leq \bar{k}_0$ as follows.
	Suppose, on the contrary, that
	$q(x_k) \not\in \mathbf{B}_{\tilde m}$ for all $k = 0,\dots, \bar{k}_0$.
	If $x_k \in \mathbf{B}_m$, 
	we have $q(x_k) \in \mathbf{B}_{\tilde m}$ from
	$|q(x_k) - x_k| \leq 
	\Delta$.
	Therefore we have $x_k \not\in \mathbf{B}_m$ for all $k=0,\dots, \bar{k}_0$.
	However, if $x_k \not\in \mathbf{B}_m$ for all
	$k \leq \bar {k}_0$, then
	the Lyapunov function decreases as \eqref{eq:Lyapunov_decrease},
	and hence
	$V(x(\bar{k}_0)) \leq \alpha m^2$. 
	This implies that $|x(\bar{k}_0) |\leq m$, 
	which leads to a contradiction.
	
	From $q(x_k) \in \mathbf{B}_{\tilde m}$, we observe that
	$x_k \in  \mathbf{B}_{\bar m}$ and hence that $x_k \in \mathcal{E}_{\bar m}$. 
	Fig.~\ref{fig:state_trajectory} illustrates the regions 
	used in this proof.
	\begin{figure}[t]
		\centering
		\includegraphics[width = 8cm,clip]{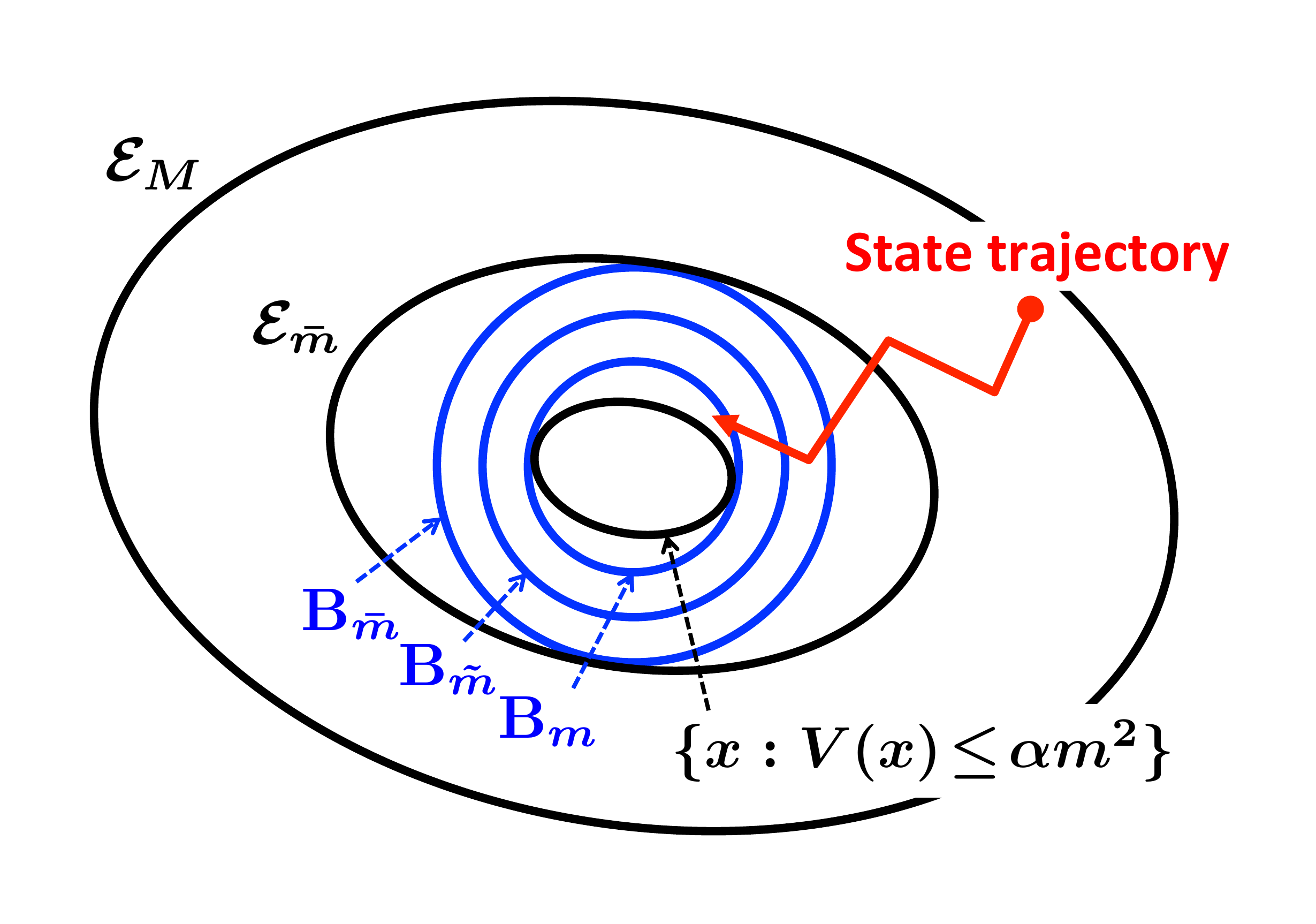}
		\caption{The regions used in the proof}
		\label{fig:state_trajectory}
	\end{figure}
	
	The invariance of $\mathcal{E}_{M} \cap \mathbf{X}$ for
	the state trajectories
	can be proved as in Theorem \ref{thm:Input_convergence}.
	This completes the proof.
\end{proof}

In the input quantization case
of Theorem \ref{thm:Input_convergence},
we use the {\em original} state for the adjustment of 
the zoom parameter $\mu$.
By contrast, in the state quantization case, 
we can achieve the asymptotic stability by
adjusting $\mu$ with the {\em quantized} state.
\begin{theorem}
	\label{thm:state_convergence}
	{\em
		Consider the PWA system \eqref{eq:state_saturation_stateEq}
		with given $K_i$ and $g_i$.
		Let Assumptions \ref{ass:Lyapunov_for_QS} and 
		\ref{ass:state_ver} hold.
		Let the initial state $x_0 \in \mathcal{E}_{M} \cap {\bf X}$ and
		the initial zoom parameter $\mu_0 = 1$.
		Assume that \eqref{eq:Km_M_state} and \eqref{eq:invariant_cond_state_ver} 
		hold, and define
		\begin{equation}
		\label{eq:Omega_def_state_case}
		\Omega := \sqrt{\frac{\beta}{\alpha}} \cdot
		\frac{\bar m}{M} < 1.
		\end{equation}
		Adjust $\mu$ by $\mu_{k} = \Omega\mu_{k-1}$
		when $q_{\mu_{k-1}}(x_k)$ gets to 
		$\mathbf{B}_{\mu_{k-1} \tilde{m}}$,
		where $\tilde m := m + \sqrt{\sf n}\Delta$,
		and send to the controller the quantized state $q_{\mu_k}(x_k)$ at time $k$.
		This event-based update strategy of $\mu$
		leads to $x_k \to 0~(k \to \infty)$.
	}
\end{theorem}
\begin{proof}
	If we observe $q_{k-1}(x_k) \in \mathbf{B}_{\mu_{k-1}\tilde m}$ 
	at time $k$,
	then $x_k \in \mathbf{B}_{\mu_{k-1}\bar m}$,
	where $\bar m := m + 2 \sqrt{\sf n} \Delta$. 
	Hence we obtain $x_k \in 
	\mathcal{E}_{\mu_kM}$ after the update $\mu_k = \Omega \mu_{k-1}$.
	The other part of the proof follows in the same line 
	as that of Theorem \ref{thm:Input_convergence},
	so we omit it.
\end{proof}

\begin{remark}
	Another approach to stabilize the PWA system with the quantized state 
	feedback is
	to combine the plant and the quantizer.
	In this case, we consider the following PWA system:
	\begin{align}
	x_{k+1} 
	&= A_ix_k + B_i u_k + f_i& \quad &(x_k  \in \mathcal{X}_i) \notag \\
	y_k &= q_j& \quad &(x_k  \in \mathcal{Q}_j) \label{eq:PWA_Quantizer} \\
	u_k &= K_i y_k + g_i& \quad &(x_k  \in \mathcal{X}_i). \notag
	\end{align}
	The difficulty of this approach is that
	we need to stabilize PWA systems with output feedback $y_k = q_j$.
	Output feedback stabilization of PWA systems has been
	studied in \cite{Qiu2011} and the reference therein, but
	the output structure in these previous works is $y_k = C_i x_k$.
	In general, it is difficult to design stabilizing controllers 
	for the system \eqref{eq:PWA_Quantizer}.
	Moreover, if we adjust the quantizer, then the system 
	\eqref{eq:PWA_Quantizer} becomes time varying.
	To avoid technical issues, we do not proceed along this
	lines.
\end{remark}

\subsection{Strategy in Controller}
As in \cite[Section 7.2]{Liberzon2014},
a better quantization value can be computed 
in the controller side if the state is near switching boundaries.
For the recompution of a new quantization value, we make the following assumption:
\begin{assumption}
	\label{ass:QRP}
	{\em
		The controller has the information on 
		the switching regions $\{\mathcal{X}_i\}_{i \in \mathcal{S}}$.
		All quantization regions $\mathcal{Q}_j$ are polyhedra.
	}
\end{assumption}

If the quantized state $q(x_k)$ is in a quantization region that has no
switching boundary, then the controller uses $q(x_k)$.
On the other hand, 
in order to achieve better performance, 
if the corresponding quantization region contains a switching boundary,
then the controller can generate a new quantized value from the information on
the quantized state and
the currently active mode as follows.

Let the switching region corresponding to the active mode 
be $\mathcal{X}_i$ and let
the quantization region of the transmitted quantized state
be $\mathcal{Q}_j$.
Then the state belongs to $\mathcal{X}_i \cap \mathcal{Q}_j$.
Suppose that $\mathcal{X}_i \cap \mathcal{Q}_j$
is bounded. Otherwise, the controller does not
recompute a new quantization value.
Since both regions are polyhedra,
$\mathcal{X}_i \cap \mathcal{Q}_j$ is a polyhedron.
Let us denote its closure by $\mathcal{A}$ .

Since $x \in \mathcal{A}$,
the controller computes a new quantized state 
\begin{equation*}
q_{\text{new}} := 
\argmin_{\xi \in \mathbb{R}^{\sf n}} \max_{x \in \mathcal{A}}
|\xi - x|_{\infty},
\end{equation*}
which is the Chebyshev center of $\mathcal{A}$.

The next theorem shows that $q_{\text{new}}$ can be obtained by
linear programming
and that the quantization error by using 
$q_{\text{new}}$ as the new quantized state 
is always less than
or equal to the quantization level $\Delta$ in \eqref{eq:quantizer_cond1}.
\begin{theorem}
	\label{thm:new_quntize}
	{\em
		Let the vertices of $\mathcal{A}$ be $v_1,\dots,v_{\ell}$.
		The new quantization value $q_{\text{\em new}}$ 
		is computed by the following linear program:
		\begin{align}
		&\text{Minimize $\delta \geq 0$ such that there exists
			$\xi \in \mathbb{R}^{\sf n}$ satisfying} \notag \\
		&~~ \text{$\xi - v_i \leq \delta {\bf 1}$ and 
			$\xi - v_i \geq - \delta {\bf 1}$ for all $i=1,\dots,\ell$.}
		\label{eq:q_new_LP}
		\end{align}
		Moreover, if $|x| < M$, then $q_{\text{\em new}}$ satisfies 
		\[
		\max_{x \in \mathcal{A}} |q_{\text{\em new}} - x|_{\infty}
		\leq \Delta.
		\]
	}
\end{theorem}
\begin{proof}
	It is well known that 
	for every $\xi \in \mathbb{R}^n$,
	$\max_{x \in \mathcal{A}}
	|\xi - x|_{\infty} = \max_{z \in \{v_1,\dots,v_{\ell}\}}
	|\xi - x|_{\infty}$; see also Appendix.
	Hence the linear program \eqref{eq:q_new_LP} gives $q_{\text{new}}$.
	
	Since $\mathcal{A} \subset \cl ( \mathcal{Q}_j)$, it follows 
	from \eqref{eq:quantizer_cond1}
	that if $|x| \leq M$, then
	\begin{align*}
	\max_{x \in \mathcal{A}} |q_{\text{new}} - x|_{\infty}
	&=
	\min_{\xi \in \mathbb{R}^{\sf n}}
	\max_{x \in \mathcal{A}} |\xi - x|_{\infty} \\
	&\leq
	\min_{\xi \in \mathbb{R}^{\sf n}}
	\max_{x \in  \cl ( \mathcal{Q}_j)} |\xi - x|_{\infty} \\
	&\leq
	\max_{x \in  \cl ( \mathcal{Q}_j)} |q_j- x|_{\infty}
	\leq \Delta,
	\end{align*}
	where $q_j$ is the original quantization value of $\mathcal{Q}_j$.
\end{proof}
\begin{remark}
	{\bf (a)}
	If the original quantization region $\mathcal{Q}_j$ is a polyhedron, then
	the zoomed-in quantization region 
	$\{ x \in \mathbb{R}^{\sf n} :~ q_{\mu}(x) = \mu q_j\}$
	is also a polyhedron.
	We can therefore compute the new quantization value $q_{\text{new}}$
	after adjusting the zoom parameter $\mu$ as well.
	
	{\bf (b)}
	The use of $q_{\text{new}}$
	does not affect the stability analysis in Theorems \ref{thm:Mtom}
	and \ref{thm:state_convergence},
	because its quantization error does not exceed $\Delta$.
	To obtain $q_{\text{new}}$, 
	we need to solve the linear program \eqref{eq:q_new_LP}.
	If the computation is not finished by
	the time when the control input is generated,
	then the controller can use the original quantization value $q_j$.
\end{remark}

\section{Controller Synthesis for PWA systems with Bounded Disturbance}
For quantized control,
here we aim to find a feedback gain $K_i$ 
and an affine term $g_i$ satisfying
\eqref{eq:Lyapunov_bound} and \eqref{eq:diff_VjVi} for every $i \in \mathcal{S}$, 
$j \in \mathcal{S}_i$, and $x \in \mathcal{X}_i$.
To this effect, we show how to obtain a set containing 
$S_i$ in \eqref{eq:Sei_def}
with less conservatism.

\subsection{Difficulty of controller synthesis for PWA systems}
Let us consider discrete-time PWA systems \eqref{eq:PWAS}
with affine state feedback control \eqref{eq:input} under
no quantization.
Theorem 1 in \cite{Ferrari2002} shows that
in order to stabilize the PWA system \eqref{eq:PWAS},
it is enough to find a feedback gain
$K_i$ and an affine term $g_i$ for every $i \in \mathcal{S}$ such that
$(A_i + B_iK_i)x + f_i+B_ig_i 
\in \mathbf{X}$ ($x \in \mathcal{X}_i$) and
the piecewise Lyapunov function 
$V(x)$
satisfies \eqref{eq:Lyapunov_bound} and
\begin{gather}
V((A_i \!+\! B_iK_i)x \!+\! f_i+B_ig_i) \!-\! V(x) \leq -\gamma  |x|^2
\quad (x \in \mathcal{X}_i) \label{eq:Lyapunov_diff}
\end{gather}
for some
$\alpha, \beta, \gamma  >0$.

Define $V(x) := V_i(x)$ ($x \in \mathcal{X}_i$),
with a function $V_i:\mathcal{X}_i \to \mathbb{R}$.
The sufficient condition of \eqref{eq:Lyapunov_diff}
used for the stability analysis in \cite{Ferrari2002, Feng2002} is that
\begin{equation}
\label{eq:diff_VjVi_difficulty}
V_j((A_i + B_i K_i)x+f_i+B_ig_i) - V_i(x) \leq -\gamma |x|^2
\end{equation}
for all $x \in \mathcal{X}_i$ and $j \in \mathcal{S}$ with
$\mathcal{X}_j \cap \mathcal{R}_i \not = 0$, where
$\mathcal{R}_i$ is the one-step reachable set defined in \eqref{eq:reachable_set_def}.
However, 
it is generally difficult to obtain 
$K_i$ and $g_i$ satisfying this condition in a less conservative way.
This is because $j$, namely, 
the polyhedron to which $(A_i + B_i K_i)x+f_i+B_ig_i$ may belong is
dependent of the unknown variables $K_i$, $g_i$.
To circumvent this difficulty, 
it is assumed, e.g., in \cite{Cuzzola2002, Lazar2008, Xu2014} that
the state can reach every polyhedron in one step, but 
this assumption makes the controller synthesis conservative 
if disturbances are bounded. In addition to that, checking 
the condition \eqref{eq:diff_VjVi_difficulty} for
every pair $(i,j)$ leads to
computational complexity for PWA systems with large number of modes.
Therefore the objective here is to obtain a set to which 
the state go in one step under bounded disturbance.

\subsection{One-step reachable set for PWA systems with bounded disturbances}
Consider 
a PWA system with bounded disturbances given by
\begin{align}
x_{k+1} 
&= A_ix_k + B_i K_i x_k + f_i + B_ig_i + D_i d_k 
\qquad
(x_k  \in \mathcal{X}_i) \notag \\
&= G(x_k,K_ix_k+g_i) + D_i d_k, 
\label{eq:PWA_disturbance}
\end{align}
where the disturbance $d_k$ satisfies
$d_k \in \mathbf{B}_{\Delta}^{\infty}
=
\{d \in \mathbb{R}^{\sf d}:~ |d|_{\infty} \leq \Delta\}$
for all $k \geq 0$.
The next lemma gives a motivation of studying the set $\mathcal{S}_i$ 
defined in \eqref{eq:Sei_def} in terms of practical input-state-stability
in addition to quantized control in the previous sections.
A proof is provided for completeness.
\begin{lemma}
	{\it
		Let $\Delta > 0$.
		Define 
		$\mathcal{R}_i := \{ G_i(x,K_ix+g_i): x\in \mathcal{X}_i\}$ and
		$\mathcal{B}_{i} := \{
		D_i d : |d|_{\infty} \leq \Delta
		\}$.
		For every $i \in \mathcal{S}$, assume that
		$
		\mathcal{R}_i \oplus \mathcal{B}_i
		\subset {\bf X}
		$.
		If the piecewise Lyapunov function 
		$V(x) := V_i(x)$ ($x \in \mathcal{X}_i$),
		with a function $V_i:\mathcal{X}_i \to \mathbb{R}$, satisfies 
		\eqref{eq:Lyapunov_bound} for some $\alpha , \beta > 0$
		and there exist
		$\gamma > 0$ and $\rho > 0$ such that
		for every $i \in \mathcal{S}$ and
		$j \in \mathcal{S}_i$ and for every $x \in \mathcal{X}_i$ and 
		$d \in \mathbf{B}_{\Delta}^{\infty}$,
		\begin{equation}
		\label{eq:diff_VjVi_ISS}
		V_j((A_i+B_iK_i)x+f_i+B_ig_i + D_i d) - V_i(x) \leq -\gamma |x|^2 + 
		\rho \Delta^2,
		\end{equation}
		then we have 
		\begin{equation}
		\label{eq:ISSprop}
		|x_k|^2 \leq 
		\frac{\beta}{\alpha} (1-\epsilon)^k |x_0|^2
		+ \frac{\rho}{\alpha \epsilon} \Delta^2,
		\end{equation}
		where 
		$\epsilon :=  \gamma / \beta$.
	}
\end{lemma}
\begin{proof}
	Since $x_{k+1} \in \mathcal{R}_i \oplus \mathcal{B}_i
	\subset {\mathbf{X}}$
	for all $x_k \in \mathcal{X}_i$, it follows that if $x_k \in \mathcal{X}_i$,
	then $x_{k+1} \in \mathcal{X}_j$ for some $j \in \mathcal{S}_i$.
	Therefore \eqref{eq:Lyapunov_bound} and \eqref{eq:diff_VjVi_ISS} give
	\begin{align*}
	V(x_{k+1}) \leq (1-\epsilon) V(x_{k}) + \rho \Delta^2,
	\end{align*}
	and hence
	\begin{align}
	\label{eq:xk_bound}
	V(x_k) \leq (1-\epsilon)^k V(x_0) + \frac{\rho}{\epsilon} \Delta^2.
	\end{align}
	Using \eqref{eq:Lyapunov_bound} again, we obtain \eqref{eq:ISSprop}
	from \eqref{eq:xk_bound}.
\end{proof}

\subsubsection{One-step reachable set with known $K_i$ and $g_i$}
First we study the case when $K_i$ and $g_i$ are known.
The lemma below gives a condition equivalent to 
$\mathcal{X}_j  \cap \bigl(\mathcal{R}_i \oplus \mathcal{B}_{i}\bigr)
\not = \emptyset$ in the definition \eqref{eq:Sei_def}
of $\mathcal{S}_i$.
\begin{lemma}
	\label{lem:min_sum_trans}
	{\em
		Define $
		\mathcal{B} := \{
		D d : |d|_{\infty} \leq \Delta
		\}$.
		For arbitrary sets $\mathcal{M}_1, \mathcal{M}_2 \subset \mathbb{R}^{\sf n}$,
		we have
		\begin{equation*}
		\left(\mathcal{M}_1 \oplus \mathcal{B}\right) \cap \mathcal{M}_2 \not= \emptyset
		\quad \Leftrightarrow \quad
		\mathcal{M}_1 \cap \left( \mathcal{M}_2 \oplus \mathcal{B} \right) \not=\emptyset.
		\end{equation*}
	}
\end{lemma}
\vspace{6pt}
\begin{proof}
	It suffices to show that if there exists $\xi_1 \in \mathbb{R}^{\sf n}$ satisfying
	$\xi_1 \in \left(\mathcal{M}_1 \oplus \mathcal{B} \right) \cap \mathcal{M}_2$,
	then we have $\xi_2 \in \mathbb{R}^{\sf n}$ such that
	\begin{equation}
	\label{eq:B_transit_conclusion}
	\xi_2 \in \mathcal{M}_1 \cap \left( \mathcal{M}_2 \oplus \mathcal{B} \right).
	\end{equation}
	
	Since $\xi_1 \in \left(\mathcal{M}_1 \oplus \mathcal{B}\right) \cap \mathcal{M}_2$,
	it follows that
	$
	\xi_1 = m_1 + Dd
	$
	for some $m_1 \in \mathcal{M}_1$ and for some
	$d \in \mathbf{B}_{\Delta}^{\infty}$,
	and also that $\xi_1 \in \mathcal{M}_2$.
	Moreover, since $-Dd \in \mathcal{B}$, we have
	\[
	m_1 = \xi_1 - Dd \in \mathcal{M}_2 \oplus \mathcal{B}.
	\]
	The desired conclusion \eqref{eq:B_transit_conclusion} 
	holds with $\xi_2 = m_1$.
\end{proof}

We see from Lemma \ref{lem:min_sum_trans} that
$\mathcal{X}_j \cap
\bigl(\mathcal{R}_i \oplus \mathcal{B}_i \bigr)
\not = 0$ is equivalent to
$\mathcal{R}_i 
\cap \bigl(\mathcal{X}_j  \oplus \mathcal{B}_i \bigr) \not = \emptyset$.
Therefore $\mathcal{S}_i$ in \eqref{eq:Sei_def}
satisfies
\begin{equation*}
\mathcal{S}_i =
\left\{
j \in \mathcal{S}:~
\mathcal{R}_i 
\cap \bigl(\mathcal{X}_j  \oplus \mathcal{B}_i\bigr) 
\not = \emptyset
\right\}.
\end{equation*} 

The following theorem gives a set containing $\mathcal{S}_i$,
which can be obtained by linear programing:
\begin{theorem}
	\label{lem:Set_containing_Si}
	{\em
		Using suitable $U_i$ and $v_i$, we can write
		the closure of $\mathcal{X}_i$ as
		\begin{equation}
		\label{eq:X_j_rep}
		\cl (\mathcal{X}_i) = \{ x:~U_i x \leq v_i\}
		\qquad (i \in \mathcal{S}).
		\end{equation}
		Define $\mathcal{S}_i$ as in \eqref{eq:Sei_def}.
		If we define $\bar{\mathcal{S}}_i$ by
		\begin{align}
		\bar{\mathcal{S}}_i &:=
		\bigl\{
		j \in \mathcal{S}:~
		U_i x \leq v_i,~~
		d \leq \Delta {\bf 1},~~
		d \geq -\Delta {\bf 1}, \notag \\
		&\qquad \qquad 
		\text{and~~}
		U_j((A_i+B_iK_i) x + f_i + B_ig_i - D_i d) \leq v_j \notag \\
		&\qquad \qquad 
		\text{for some~~} x \in \mathbb{R}^{\sf n}
		\text{~and~} 
		d \in \mathbb{R}^{\sf d}
		\bigr\}
		\label{eq:bSei}
		\end{align}
		then $\mathcal{S}_i \subset \bar{\mathcal{S}}_i$.
	}
\end{theorem}

\begin{proof}
	First of all, we see that
	there exists $x \in \mathbb{R}^{\sf n}$ satisfying both $x \in \mathcal{R}_i$ and 
	$x \in \mathcal{X}_j \oplus \mathcal{B}_{i}$ if and only if
	there exists $x \in \mathcal{X}_i$ such that 
	$\bar{A}_i x + \bar f_i \in 
	\mathcal{X}_j \oplus \mathcal{B}_{i}$, where
	$\bar A_i := A_i + B_i K_i$ and $\bar f_i := f_i + B_ig_i$.
	
	By definition,
	$\bar{A}_i x + \bar f_i \in 
	\cl (\mathcal{X}_j) \oplus \mathcal{B}_{i}$ is equivalent to
	\begin{equation*}
	\bar{A}_i x+ \bar f_i = z + D_id
	\end{equation*}
	for some $z \in \mathbb{R}^{\sf n}$ and 
	$d \in \mathbb{R}^{\sf d}$ satisfying
	$U_jz \leq v_j$ and $|d|_{\infty} \leq \Delta$.
	Therefore $\bar{A}_i x + \bar f_i \in 
	\cl (\mathcal{X}_j) \oplus \mathcal{B}_{i}$ 
	is equivalent to
	\begin{align*}
	d \leq \Delta {\bf 1},~
	d \geq -\Delta {\bf 1},~\text{and}~
	U_j(\bar{A}_i x + \bar f_i - D_id) \leq v_j
	\end{align*}
	for some $d \in \mathbb{R}^{\sf d}$.
	
	Thus
	we obtain
	the following fact: If $\mathcal{R}_i 
	\cap \bigl(\mathcal{X}_j  \oplus \mathcal{B}_{i}\bigr) 
	\not= \emptyset$, then
	\begin{align}
	\label{eq:BD_NSC}
	\mathcal{X}_i \cap 
	\bigr\{x \in \mathbb{R}^{\sf n}:~
	d \leq \Delta {\bf 1},~
	d \geq -\Delta {\bf 1}, \text{~and~}
	U_j(\bar{A}_i x + \bar f_i - D_id) \leq v_j
	~~\text{for some $d \in \mathbb{R}^{\sf d}$}
	\bigl\} \not= \emptyset.
	\end{align}
	Noticing that $j \in \mathcal{S}$ satisfies \eqref{eq:BD_NSC} if and only if
	$j \in \bar{\mathcal{S}}_i$, 
	we have that $\mathcal{S}_i \subset \bar{\mathcal{S}}_i$.
\end{proof}

The conservatism of Theorem \ref{lem:Set_containing_Si} is due to 
only $\mathcal{X}_j \subset \cl (\mathcal{X}_j)$.
If we allow more conservative results, then we can use the set
$\tilde{\mathcal{S}}_i \supset 
\mathcal{S}_i$ below, which
can be obtained with less computational cost
by removing the disturbance term $d$.
A similar idea is used for the analysis of reachability with bounded disturbance
in \cite{Lin_reachability2011}.
\begin{corollary}
	\label{cor:tilde_Sei}
	{\it
		Let $\bar u^{(l)}_{ji}$ be the sum of the absolute value of the elements in
		$l$-th row
		of $U_j D_i$ and define $\bar v_{ji} := [\bar{v}_{j,i}^{(1)}~\dots~
		\bar{v}_{j,i}^{({\sf n}_{U})}]^{\top}$,
		where ${\sf n}_{U}$ is the number of rows of $U_jD_i$.
		If we define
		$\tilde{\mathcal{S}}_i$ by
		\begin{align}
		\tilde{\mathcal{S}}_i 
		:=
		\bigl\{
		j \in \mathcal{S}:~
		U_i x \leq v_i \text{~~and~~} 
		U_j((A_i+B_iK_i) x + f_i ) \leq v_j + \Delta \bar v_{ji}
		~~~\text{for some~} x \in \mathbb{R}^{\sf n}
		\bigr\},
		\label{eq:tilSej}
		\end{align}
		then $\mathcal{S}_i$ in
		\eqref{eq:Sei_def} satisfies
		$\mathcal{S}_i \subset \tilde{\mathcal{S}}_i$.
	}
\end{corollary}

\begin{proof}
	It suffices to prove that
	\begin{equation}
	\label{eq:Xje_containing}
	\mathcal{X}_j \oplus \mathcal{B}_{\Delta} \subset
	\{
	x\in \mathbb{R}^{\sf n}:~ U_jx \leq v_j + \Delta \bar v_{ji}
	\}.
	\end{equation}
	Indeed, if \eqref{eq:Xje_containing} holds, then
	$\mathcal{R}_i 
	\cap \bigl(\mathcal{X}_j  \oplus \mathcal{B}_{i}\bigr) 
	\not= \emptyset$ implies
	\begin{equation*}
	\mathcal{X}_i \cap 
	\{x\in \mathbb{R}^{\sf n} :~
	U_j(\bar{A}_i x + \bar f_i ) \leq v_j + \Delta\bar v_{ji}
	\} \not= \emptyset,
	\end{equation*}
	where ${\bar A}_i := A_i +B_iK_i$ and $\bar f_i = f_i +B_ig_i$.
	This leads to $\mathcal{S}_i \subset \tilde{\mathcal{S}}_i$.
	
	Let us study the first element of $U_j(x+D_id)$.
	Let $U_{j}^{(1,l)}$, $(U_jD_i)^{(1,l)}$, and $v_{j}^{(1)}$ 
	be the $(1,l)$-th entry of $U_j$, $U_jD_i$ and the first entry of $v_j$, respectively.
	Also let $x^{(l)}$ and $d^{(l)}$ be
	the $l$-th element of $x$ and $d$, respectively.
	If $x \in \cl (\mathcal{X}_j) $ and $d \in \mathbf{B}_{\Delta}^{\infty}$, then
	the first element $\xi_{ji}^{(1)}$ of $U_j(x+D_id)$ satisfies
	\begin{align}
	\xi_{ji}^{(1)} = \sum_{l=1}^{\sf n} 
	\left(U_j^{(1,l)}x^{(l)} + (U_jD_i)^{(1,l)} d^{(l)}
	\right)
	\leq 
	v_j^{(1)} + \sum_{l=1}^{\sf n} (U_jD_i)^{(1,l)} d^{(l)} 
	\leq
	v_j^{(1)} +  \Delta \bar{v}_{ji}^{(1)}.
	\label{eq:bound_dis_Poly}
	\end{align}
	Since we have the same result for the other elements of $U_j(x+D_id)$,
	it follows that \eqref{eq:Xje_containing} holds.
\end{proof}


\subsubsection{One-step reachable set with unknown $K_i$ and $g_i$}
Let us next investigate the case when $K_i$ and $g_i$ are unknown.

The set $\bar{\mathcal{S}}_i$ given in
Theorem \ref{lem:Set_containing_Si}
works 
for stability analysis
in the presence 
of bounded disturbances, but $\bar{\mathcal{S}}_i$
is dependent on the feedback gain
$K_i$ and the affine term $g_i$.
Hence we cannot use it for their design. 
Here we obtain a set $\mathcal{T}_i \supset \mathcal{S}_i$, 
which does not depend on $K_i$, $g_i$.
Moreover, we derive a sufficient condition on $K_i$, $g_i$
for the state to belong to a given polyhedron in one step. 

Let
$\mathbf{U}$ be the polyhedron defined by
\begin{equation*}
\mathbf{U} := \{u \in \mathbb{R}^{\sf m}:~ Ru \leq r  \},
\end{equation*}  
and we make an additional constraint that 
$u_k \in \mathbf{U}$ for all $k \geq 0$.
Similarly to \cite{Bemporad2000HSCC},
using the information on the input matrices $B_i$ and 
the input bound $\mathbf{U}$, 
we obtain a set
independent of $K_i$, $g_i$ to which the state belong in one step.
\begin{theorem}
	\label{coro:Si_sufficient_with_disturbance}
	{\it
		Assume that for each $i \in \mathcal{S}$,
		$K_i \in \mathbb{R}^{{\sf n} \times {\sf m}}$ and 
		$g_i \in \mathbb{R}^{\sf m}$
		satisfy
		$(A_i+B_iK_i)x + f_i + B_ig_i + D_id \in \mathbf{X}$
		and $K_ix +g_i \in \mathbf{U}$
		if $x \in \mathcal{X}_i$ and $d \in \mathbf{B}_{\Delta}^{\infty}$.
		Let the closure of $\mathcal{X}_i$ be given by \eqref{eq:X_j_rep}.
		Define
		\begin{align}
		\mathcal{T}_i &:=
		\bigl\{
		j \in \mathcal{S}:~
		U_i x \leq v_i,~~Ru \leq r,~~
		d \leq \Delta {\bf 1},~~
		d \geq -\Delta {\bf 1}, \notag \\
		&\qquad
		\text{and~~}
		U_j(A_i x + B_i u + f_i + D_i d) \leq v_j \notag \\
		&\qquad
		\text{for some~~} x \in \mathbb{R}^{\sf n},~
		u \in \mathbb{R}^{\sf m},
		\text{~and~} 
		d \in \mathbb{R}^{\sf d}
		\bigr\}.
		\label{eq:bSi_def}
		\end{align}
		Then 
		we have
		\begin{equation}
		\label{eq:Gi_in_Xs}
		(A_i+B_iK_i)x + f_i+B_ig_i+D_id \in \sum_{j \in \mathcal{T}_i} \mathcal{X}_j
		\end{equation}
		for all $x \in \mathcal{X}_i$
		and $d \in \mathbf{B}_{\Delta}^{\infty}$,
		and hence $\mathcal{S}_i$ in \eqref{eq:Sei_def}
		satisfies $\mathcal{S}_i \subset \mathcal{T}_i$.
	}
\end{theorem}
\begin{proof}
	Define $G_i(x) := (A_i+B_iK_i)x + f_i+B_ig_i$.
	To show \eqref{eq:Gi_in_Xs},
	it suffices to prove that 
	for all $x \in \mathcal{X}_i$ and
	$d \in \mathbf{B}_{\Delta}^{\infty}$, there exists 
	$j \in \mathcal{T}_i$ such that 
	$G_i(x)  + D_i d
	\in \mathcal{X}_j$.
	
	Suppose, on the contrary, that
	there exist $x \in \mathcal{X}_i$ and
	$d \in \mathbf{B}_{\Delta}^{\infty}$
	such that
	$G_i(x) + D_i d \not\in \mathcal{X}_j$ 
	for every $j \in \mathcal{T}_i$.
	Since $G_i(x) +D_id \in \mathbf{X}$, 
	it follows that $G_i(x) +D_id \in \mathcal{X}_{j}$
	for some $j \in \mathcal{S}$. Also, by definition
	\begin{align*}
	&\mathcal{T}_i =
	\bigl\{
	j \in \mathcal{S}:~
	A_ix +B_iu + f_i + D_i d \in \cl(\mathcal{X}_j)
	~~\text{for some
		$x \!\in \! \cl(\mathcal{X}_i)$, $u \!\in \! \mathbf{U}$, 
		and $d \! \in \! \mathbf{B}_{\Delta}^{\infty}$}
	\bigr\}.
	\end{align*}
	Since $x \in \mathcal{X}_i$, 
	$u = K_ix +g_i \in \mathbf{U}$,
	$d \in \mathbf{B}_{\Delta}^{\infty}$, and
	$G_i(x)+D_id \in \mathcal{X}_j$, 
	it follows that $j \in \mathcal{T}_i$. Hence we have 
	$G_i(x)+D_id \in \mathcal{X}_j$ for some $j \in \mathcal{T}_i$.
	Thus we have a contradiction and
	\eqref{eq:Gi_in_Xs} holds
	for every $x \in \mathcal{X}_i$
	and $d \in \mathbf{B}_{\Delta}^{\infty}$.
	
	Let us next prove $\mathcal{S}_i \subset \mathcal{T}_i$.
	Let $j \in \mathcal{S}_i$. By definition,
	there exists $x \in \mathcal{X}_i$ and $d \in \mathbf{B}_{\Delta}^{\infty}$
	such that $G_i(x) + D_id \in \mathcal{X}_j$.
	Also, we see from \eqref{eq:Gi_in_Xs} that
	there exists $\bar j \in \mathcal{T}_i$ such that
	$G_i(x) + D_id \in \mathcal{X}_{\bar j}$.
	Hence we have $\mathcal{X}_{j} \cap \mathcal{X}_{\bar j}  \not= \emptyset$,
	which implies $j = \bar j$. Thus we have $\mathcal{S}_i \subset \mathcal{T}_i$.
\end{proof}

See Remark \ref{rem:Total_ss_cond}
for the assumption that 
$(A_i+B_iK_i)x + f_i +D_id \in \mathbf{X}$
and $K_ix+g_i \in \mathbf{U}$ for all $x \in \mathcal{X}_j$.
\begin{remark}
	{\bf (a)}
	In 
	Theorem \ref{coro:Si_sufficient_with_disturbance}, we have used 
	the counterpart of 
	$\bar{\mathcal{S}}_i$ given in Theorem \ref{lem:Set_containing_Si}, but
	one can easily modify the theorem based on
	$\tilde{\mathcal{S}}_i$ in Corollary \ref{cor:tilde_Sei}.
	
	{\bf (b)}
	If $B_i$ is full row rank, then 
	for all $x \in \cl (\mathcal{X}_i)$, $\eta \in \cl (\mathcal{X}_j)$, 
	and $d \in \mathbf{B}_{\Delta}^{\infty}$, 
	there exists $u \in \mathbb{R}^{\sf m}$
	such that $B_i u = \eta - A_ix - f_i - D_id$.
	In this case, we have the trivial fact: $\mathcal{T}_i = \mathcal{S}$.
\end{remark}

Theorem \ref{coro:Si_sufficient_with_disturbance}
ignores the affine feedback structure $u = K_ix+g_i$ ($x \in \mathcal{X}_i$), 
which makes this theorem conservative.
Since the one-step reachable set depends on the unknown parameters $K_i$
and $g_i$, 
we cannot utilize the feedback structure unless we add some conditions on
$K_i$ and $g_i$.
In the next theorem,
we derive linear programming on $K_i$ and $g_i$ for 
a bounded $\mathcal{X}_i$, 
which is a sufficient condition for 
the one-step reachable set under bounded disturbances
to be contained in a given polyhedron.
\begin{theorem}
	\label{thm:add_cond_Ki}
	{\em
		Let a polyhedron
		$\mathcal{Z} = 
		\{x \in \mathbb{R}^{{\sf n}}:
		\Phi x \leq \phi\}$, and
		let $\mathcal{X}_i$ be a bounded polyhedron.
		Let $\{\xi_{i,1},\dots, \xi_{i,L_i}\}$ and 
		$\{d_{1},\dots, d_{\eta}\}$ be the vertices of
		$\cl (\mathcal{X}_i)$ and 
		$\mathcal{B}_{\Delta}^{\infty}$,
		respectively.
		A matrix $K_i \in \mathbb{R}^{{\sf n} \times {\sf m}}$ 
		and 
		a vector $g_i \in \mathbb{R}^{\sf m}$ satisfy
		$(A_i+B_iK_i)x + f_i+B_ig_i + 
		D_id \in {\mathcal Z}$ for all $x \in \mathcal{X}_i$
		and $d \in \mathcal{B}_{\Delta}$
		if linear programming
		\begin{align}
		\label{eq:LP_cond_Kg}
		\Phi\left(
		(A_i+B_iK_i)\xi_{i,h} + f_i + B_ig_i +D_i d_{\nu}
		\right) \leq \phi
		\end{align}
		is feasible for every $h=1,\dots,L_i$ and for every 
		$\nu=1,\dots,\eta$.
	}
\end{theorem}
\begin{proof}
	Define $G_i(x) := (A_i+B_iK_i)x + f_i+B_ig_i$.
	Relying on the results \cite[Chap.~6]{Blanchini2008} 
	(see also \cite{Barmish1979, Rubagotti2013}), 
	we have
	\begin{equation*}
	\{G_i(x)+D_id:~ x \in \cl(\mathcal{X}_i),~d \in \mathcal{B}_{\Delta}\}
	= 
	\text{conv}\{G_i(\xi_{i,h})+D_id_{\nu},~
	h=1,\dots, L_i,~\nu=1,\dots,\eta \},
	\end{equation*}
	where $\text{conv}(S)$ means the convex hull of a set $S$.
	We therefore obtain $G_i(x)+D_id \in \mathcal{Z}$ for all
	$x \in \cl (\mathcal{X}_i)$ and $d \in \mathcal{B}_{\Delta}$
	if and only if
	$G_i(\xi_{i,h})+D_id_{\nu} \in \mathcal{Z}$, or \eqref{eq:LP_cond_Kg}, 
	holds for every $h=1,\dots,L_i$ and $\nu = 1,\dots, \eta$.
	Thus the desired conclusion is derived.
\end{proof}

\begin{remark}
	{\bf (a)}
	To use Theorem \ref{thm:add_cond_Ki}, we must design
	a polyhedron $\mathcal{Z}$ in advance.
	One design guideline is to take $\mathcal{Z}$ such that
	$\mathcal{Z} \subset \sum_{j \in \bar{\mathcal{T}}_i} \mathcal{X}_j$
	for some $\bar{\mathcal{T}}_i \subset \mathcal{T}_i$, 
	where
	$\mathcal{T}_i$ is defined in Theorem \ref{coro:Si_sufficient_with_disturbance}.
	
	\noindent
	{\bf (b)}
	As in Theorem \ref{lem:Set_containing_Si},
	the conservatism in Theorem \ref{thm:add_cond_Ki}
	arises only from $\mathcal{X}_i \subset \cl (\mathcal{X}_i)$.
	
	\noindent
	{\bf (c)}
	Theorem \ref{thm:add_cond_Ki} gives a trade-off on computational complexity:
	In order to reduce the number of pairs such that  
	\eqref{eq:diff_VjVi} holds, 
	we need to solve the linear programming problem \eqref{eq:LP_cond_Kg}.
	
	\noindent
	{\bf (d)}
	When the state is quantized, then $D_i = B_iK_i$ in \eqref{eq:PWA_disturbance},
	and hence $D_i$ depends on $K_i$ linearly.
	In this case, however, Theorem \ref{thm:add_cond_Ki} can be 
	used for the controller design.
	
\end{remark}

\begin{remark}
	\label{rem:Total_ss_cond}
	Assumptions \ref{ass:input_ver}, \ref{ass:state_ver} and
	Theorem \ref{coro:Si_sufficient_with_disturbance} require
	conditions on $K_i$ and $g_i$ that
	$G_i(x,K_i x)+D_id \in {\mathbf X}$ 
	and $K_ix +g_i \in \mathbf{U}$
	for all $x \in \mathcal{X}_i$ and all $d \in \mathbf{B}_{\Delta}^{\infty}$.
	If ${\mathbf X} = \mathbb{R}^{\sf n}$ and 
	$\mathbf U = \mathbb{R}^{\sf m}$, then
	these conditions always hold.
	If ${\mathbf X}\not= \mathbb{R}^{\sf n}$ but if $\mathcal{X}_i$ is a 
	bounded polyhedron,
	then Theorem \ref{thm:add_cond_Ki} gives 
	linear programming that is sufficient for 
	$G_i(x,K_i x)+D_id \in {\mathbf X}$ to hold.
	Also, Theorem \ref{thm:add_cond_Ki} with 
	$A_i = D_i = 0$, $B_i = I$, and $f_i = 0$
	can be applied to
	$K_ix +g_i \in \mathbf{U}$.
	If $G_i(x,K_i x)+D_id \in {\mathbf X}$ 
	and $K_ix +g_i \in \mathbf{U}$ hold for bounded ${\mathbf X}$
	and $\mathbf{U}$,
	then we can easily set the quantization parameter $M$ in 
	\eqref{eq:quantizer_cond1} 
	to avoid quantizer saturation.
	Similarly, we can use Theorem \ref{thm:add_cond_Ki}
	for constraints on the state and the input.
\end{remark}

By Theorems 
\ref{coro:Si_sufficient_with_disturbance}
and \ref{thm:add_cond_Ki}, we obtain 
linear programing on $K_i$ and $g_i$ for a set containing the
one-step reachable set under bounded disturbances.
However,
in LMI conditions 
of \cite{Cuzzola2002, Lazar2008}
for \eqref{eq:Lyapunov_bound} and
\eqref{eq:diff_VjVi},
$K_i$ is obtained via the variable transformation $K_i = Y_i Q_i^{-1}$,
where $Y_i$ and $Q_i$ are auxiliary variables.
Without variable transformation/elimination,
we obtain only BMI conditions for \eqref{eq:diff_VjVi} to hold
as in Theorem 7.2.2 of \cite{Xu2014}. 
The following theorem also gives BMI conditions on $K_i$
for \eqref{eq:Lyapunov_bound} and \eqref{eq:diff_VjVi} to hold, 
but we can apply
the cone complementary linearization (CCL) algorithm~\cite{Ghaoui1997}
to these BMI conditions:
\begin{theorem}
	\label{thm:PWA_LMI}
	{\em 
		Consider the PWA system \eqref{eq:PWA_disturbance} with
		control affine term $g_i = 0$.
		Let a matrix $E_i$ satisfy $\mathcal{X}_i \subset
		\{
		x \in \mathbb{R}^{\sf n}:~ E_i x \geq 0
		\}.$
		If $f_i = 0$ and $D_i = 0$ and if
		there exist $P_i, Q_i > 0$, $K_i$, and $M_{ij}$ with all elements
		non-negative such that
		\begin{equation}
		\label{eq:aymptotic_LMI}
		\begin{bmatrix}
		P_i - E_i^{\top} M_{ij} E_i & (A_{i}+B_iK_i)^{\top} \\
		* & Q_j
		\end{bmatrix} \succ 0,\qquad
		\begin{bmatrix}
		P_i & I \\
		I & Q_i
		\end{bmatrix} \succeq 0,
		\end{equation}
		and ${\rm trace} (P_iQ_i) = 2{\sf n}$
		hold for all $i\in \mathcal{S}$ and $j \in \mathcal{S}_i$, 
		then there exist $\alpha, \beta, \gamma_i > 0$ such that
		$V(x) := x^{\top}P_i x$ ($x \in \mathcal{X}_i$) satisfies
		\eqref{eq:Lyapunov_bound} and \eqref{eq:diff_VjVi}
		for every $i\in \mathcal{S}$, $j \in \mathcal{S}_i$, and
		$x \in \mathcal{X}_i$.
		
		Furthermore, 
		consider the case $f_i \not= 0$ and $D_i \not= 0$.
		For given $\nu_1, \nu_2 > 0$ with $\nu_1\nu_2 > 1$,
		if there exist
		$P_i, Q_i > 0$, $K_i$, and $M_{ij}$ with all elements
		non-negative such that
		\begin{equation}
		\label{eq:ISS_LMI}
		\begin{bmatrix}
		P_i - E_i^{\top} M_{ij} E_i 
		& -(A_{i}+B_iK_i)^{\top} & -(A_{i}+B_iK_i)^{\top} 
		& (A_{i}+B_iK_i)^{\top} \\
		* & \nu_1 Q_j & -Q_j & 0\\
		* & * & \nu_2 Q_j & 0 \\
		* & * & * & Q_j
		\end{bmatrix} \succ 0,\qquad
		\begin{bmatrix}
		P_i & I \\
		I & Q_i
		\end{bmatrix}
		\succeq 0,
		\end{equation}
		and ${\rm trace} (P_iQ_i) = 2{\sf n}$
		hold for all $i\in \mathcal{S}$ and $j \in \mathcal{S}_i$, 
		then there exist $\alpha, \beta, \gamma, \rho > 0$ such that
		$V(x) := x^{\top}P_i x$ ($x \in \mathcal{X}_i$) satisfies
		\eqref{eq:Lyapunov_bound} and \eqref{eq:diff_VjVi_ISS}
		for every $i\in \mathcal{S}$, $j \in \mathcal{S}_i$, and
		$x \in \mathcal{X}_i$.
	}
\end{theorem}

\begin{proof}
	Since the positive definiteness of $P_i$ 
	implies \eqref{eq:Lyapunov_bound}, 
	it is enough to show that
	\eqref{eq:aymptotic_LMI} and 
	\eqref{eq:ISS_LMI} lead to
	\eqref{eq:diff_VjVi} and \eqref{eq:diff_VjVi_ISS},
	respectively.
	
	For $P_i,Q_i > 0$ satisfying the second LMI in
	\eqref{eq:aymptotic_LMI} and \eqref{eq:ISS_LMI},
	we have $\text{trace} (P_iQ_i) \geq 2 {\sf n}$.
	Furthermore, $\text{trace} (P_iQ_i) = 2 {\sf n}$
	if and only if $P_iQ_i = I$.
	Define $\bar A_i = A_{i}+B_iK_i$.
	
	Applying the Schur complement formula to 
	the LMI condition in \eqref{eq:aymptotic_LMI},
	we have
	\begin{equation*}
	P_i - \bar A_i^{\top} P_j \bar A_i
	- E_i^{\top} M_{ij} E_i \succ 0.
	\end{equation*}
	Since $E_i x \geq 0$, there exists $\gamma_i > 0$ such that
	$V_i(x) - V_j(\bar A_i x) > \gamma_i |x|^2$ for every 
	$x \in \mathcal{X}_i$. Hence we obtain \eqref{eq:diff_VjVi}.
	
	As regards \eqref{eq:ISS_LMI}, 
	it follows from Theorem 3.1 in \cite{Lazar2008} that 
	\eqref{eq:diff_VjVi_ISS} holds for some $\gamma, \rho > 0$
	if
	\begin{equation*}
	\begin{bmatrix}
	P_i - \bar A_i^{\top}P_j\bar A_i- E_i^{\top} M_{ij} E_i 
	& -\bar A_i^{\top}P_j & -\bar A_i^{\top} P_j　\\
	* & \nu_1 P_j & -P_j\\
	* & * & \nu_2 P_j 
	\end{bmatrix}　\succ 0.
	\end{equation*}
	Pre- and post-multiplying $\text{diag}(I, P_i^{-1}, P_{i}^{-1})$
	and using the Schur complement formula, we obtain 
	the first LMI in \eqref{eq:ISS_LMI}.
\end{proof}


Since $\min (\text{trace} (P_iQ_i)) = 2 {\sf n}$,
the conditions in Theorem \ref{thm:PWA_LMI} are feasible if
the problem of minimizing $\text{trace}
\left(
\sum_{i=1}^{s} P_iQ_i
\right)$
under \eqref{eq:aymptotic_LMI}/\eqref{eq:ISS_LMI}
has a solution $2ns$.
In addition to LMIs \eqref{eq:aymptotic_LMI} and \eqref{eq:ISS_LMI},
we can consider linear programming \eqref{eq:LP_cond_Kg}
for the constraint on the one-step reachable set.
The CCL algorithm solves this constrained minimization problem. 
The CCL algorithm may not 
find the global optimal solution, but, in general,
we can solve
the minimization problem 
in a more computationally efficient way than 
the original 
non-convex feasibility problem~\cite{Oliveria1997}.

\section{Numerical Example}
Consider a PWA system in \eqref{eq:state_saturation_stateEq} with
quantized state feedback, where
\begin{gather*}
A_1 = A_3 = 
\begin{bmatrix}
0.5 & -0.4 \\
0 & 2
\end{bmatrix},\quad
A_2 = A_4  = 
\begin{bmatrix}
2 & 0 \\
-1 & 1
\end{bmatrix},\quad
A_5 = A_6  = 
\begin{bmatrix}
0.5 & -0.1 \\
1 & 2
\end{bmatrix} \\
B_1 = B_3 = B_5 = B_6 = 
\begin{bmatrix}
0 \\ 1
\end{bmatrix},\quad
B2 = B4 = 
\begin{bmatrix}
-1 \\ 0.5
\end{bmatrix},\quad
f_1=f_2=f_3=f_4=f_5=f_6 = 0.
\end{gather*}
The matrix $U_i$ and the vector $v_i$ in \eqref{eq:X_j_rep} 
characterizing
the region $\mathcal{X}_i$ are given by
\begin{gather*}
U_1 = -U_3 = 
\begin{bmatrix}
1 & -1 \\
1 & 1 \\
-1 & 0 \\
1 & 0
\end{bmatrix},\quad
U_2 = -U_4 = 
\begin{bmatrix}
1 & -1 \\
-1 & -1 \\
0 & 1
\end{bmatrix},\quad
U_5 = -U_6 = 
\begin{bmatrix}
1 & -1 \\
1 & 1 \\
-1 & 0
\end{bmatrix} \\
v_1 = v_3 = 
\begin{bmatrix}
0 \\
0 \\
1 \\
-0.3  
\end{bmatrix},\quad
v_2 = v_4 = 
\begin{bmatrix}
0 \\
0 \\
1
\end{bmatrix},\quad
v_5 = v_6 = 
\begin{bmatrix}
0 \\
0 \\
0.3
\end{bmatrix}.
\end{gather*}
Let $\mathbf{X} = \sum_{i=1}^6 \mathcal{X}_i = 
\{x \in \mathbb{R}^{2}:~ 
|x|_{\infty} \leq 1
\},
$
and let us use 
a uniform-type quantizer whose
parameters in \eqref{eq:quantizer_cond1} 
are $M = 1.5$ and $\Delta = 0.01$ .
By using Theorems \ref{thm:add_cond_Ki} and \ref{thm:PWA_LMI},
we designed feedback gains $K_i$ such that the Lyapunov function
$V(x) := x^{\top}P_i x$ ($x \in \mathcal{X}_i$) satisfies
\eqref{eq:Lyapunov_bound} and \eqref{eq:diff_VjVi}
for every $i\in \mathcal{S}$, $j \in \mathcal{S}_i$, and
$x \in \mathcal{X}_i$, and the following constraint conditions hold
\begin{align}
&x_{k+1} \in \mathcal{X}_2 \text{~for all~} x_{k} \in \mathcal{X}_1, 
\label{eq:X1_X2} \\
&x_{k+1} \in \mathcal{X}_4 \text{~for all~} x_{k} \in \mathcal{X}_3, \text{and}
\label{eq:X3_X4} \\
&x_{k+1} \in \mathbf{X} \text{~for all~} x_{k} \in \mathbf{X}.
\label{eq:Total_space}
\end{align}
The resulting $K_i$ were given by
\begin{gather*}
K_1 = K_3 = 
\begin{bmatrix}
-0.6140 &-1.6368
\end{bmatrix},\qquad
K_2 = K_4 = 
\begin{bmatrix}
1.9995 &-0.5244
\end{bmatrix} \\
K_5 = K_6 = 
\begin{bmatrix}
-0.9980 &-1.9967
\end{bmatrix},
\end{gather*}
and we obtained the decease rate $\Omega = 0.7725$ in
\eqref{eq:Omega_def_state_case}
of the ``zoom'' parameter $\mu$
with $\varepsilon_{ij} = 0.01$ and $\delta_{ij} = 0.49$.

Fig.~\ref{fig:state_trajectory_Example} shows
the state trajectories with initial states on the boundaries $x_1=1$ and
$x_2 = 1$.
We observe that all trajectories converges to the origin and that
the constraint conditions
\eqref{eq:X3_X4} and \eqref{eq:Total_space}
are satisfied in the presence of quantization errors.

\begin{figure}[t]
	\centering
	\includegraphics[width = 8cm,clip]{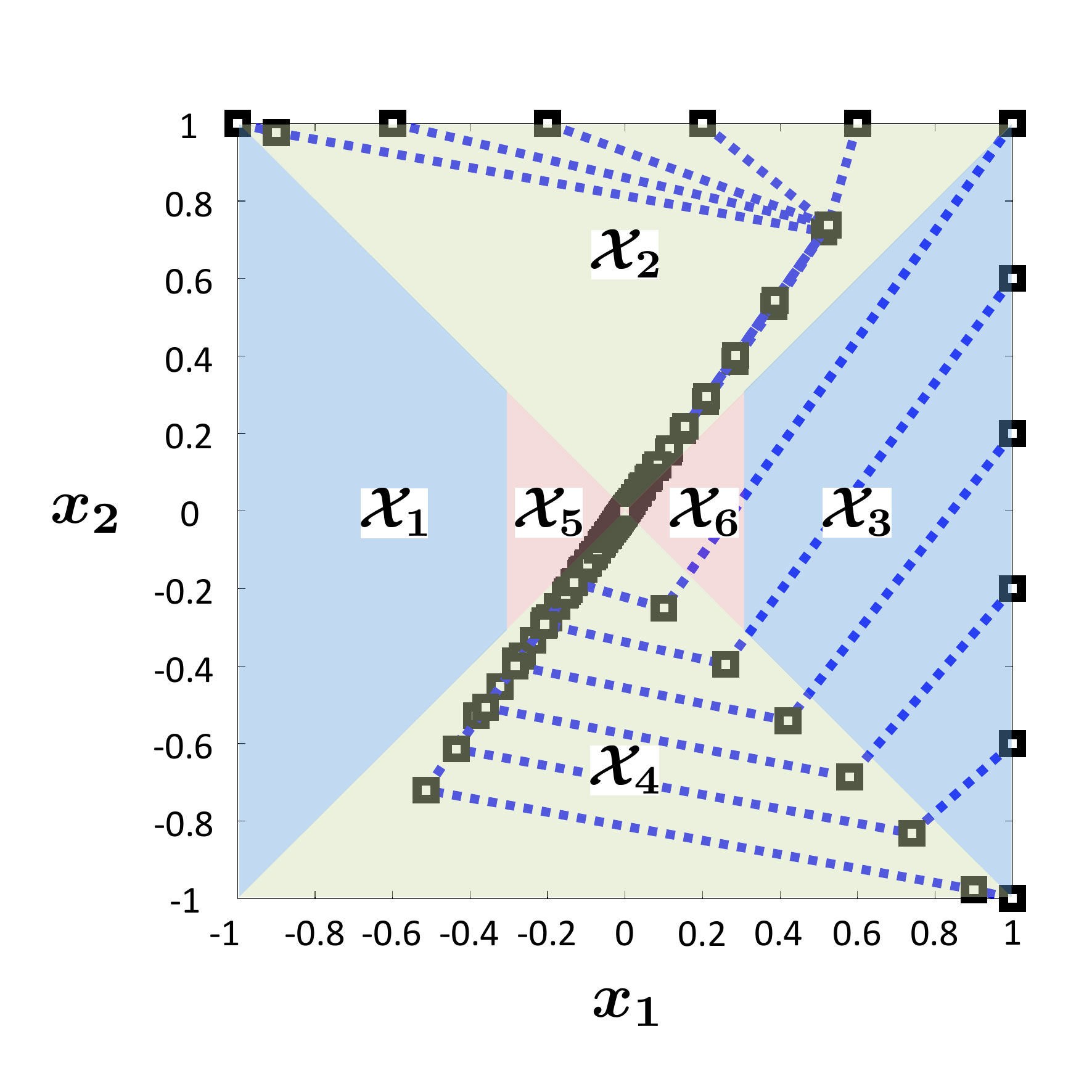}
	\caption{Simulation result}
	\label{fig:state_trajectory_Example}
\end{figure}

\section{Conclusion}
We have provided an encoding strategy for 
the stabilization of PWA systems with quantized signals.
For the stability of the closed-loop system, 
we have shown that 
the piecewise quadratic Lyapunov function
decreases in the presence of quantization errors.
For the design of quantized feedback controllers,
we have also studied the stabilization problem of
PWA systems with bounded disturbances.
In order to reduce the conservatism and the computational cost
of controller designs,
we have investigated the one-step reachable set.

\appendix
\renewcommand{\thetheorem}{\Alph{theorem}}
Here we give the proof of the following proposition
for completeness:
\begin{proposition}
	{\em 
		Let $\mathcal{A}
		\subset \mathbb{R}^{\sf n}$ be a bounded and closed polyhedron, and
		let $v_1,\dots,v_{\ell}$ be the vertices of $\mathcal{A}$.
		For every $\xi \in \mathbb{R}^{\sf n}$, 
		we have
		\begin{equation*}
		\max_{x \in \mathcal{A}}
		|\xi - x|_{\infty} = 
		\max_{x \in \{v_1,\dots,v_{\ell}\}}
		|\xi - x|_{\infty}.
		\end{equation*}
	}
\end{proposition}
\vspace{6pt}
\begin{proof}
	Choose $x \in \mathcal{A}$ arbitrarily, and let 
	\begin{align*}
	x = \sum_{p=1}^{\ell} a_pv_p,\qquad a_p \geq 0,
	\quad \sum_{p=1}^{\ell} a_p = 1.
	\end{align*}
	Let $n$-th entry of $\xi$, $x$, and $v_p$ be $\xi^{(n)}$, 
	$x^{(n)}$, and $v^{(n)}_p$, respectively.
	For every $n$, we have
	\begin{align*}
	|\xi^{(n)} \!-\! x^{(n)}| 
	\leq
	\sum_{p=1}^{\ell} a_p
	\left|
	\xi^{(n)} \!-\!  v^{(n)}_p
	\right| \! \leq \!
	\max_{n} \max_{p}\left|
	\xi^{(n)} \!-\!  v^{(n)}_p
	\right|. 
	\end{align*}
	Hence 
	$
	|\xi - x|_{\infty} 
	\leq
	\max_{n} \max_{p}|
	\xi^{(n)} -  v^{(n)}_p
	| 
	=
	\max_{p}\left|
	\xi -  v_p
	\right|_{\infty}.
	$
	This completes the proof.
\end{proof}

\section*{Acknowledgment}
The first author would like to thank 
Dr. K. Okano of University California, Santa Barbara 
for helpful discussions on quantized control for PWA systems.
The authors are also grateful to anonymous reviewers whose comments 
greatly improved this paper.




\begin{thebibliography}{10}
	\providecommand{\url}[1]{#1}
	\csname url@samestyle\endcsname
	\providecommand{\newblock}{\relax}
	\providecommand{\bibinfo}[2]{#2}
	\providecommand{\BIBentrySTDinterwordspacing}{\spaceskip=0pt\relax}
	\providecommand{\BIBentryALTinterwordstretchfactor}{4}
	\providecommand{\BIBentryALTinterwordspacing}{\spaceskip=\fontdimen2\font plus
		\BIBentryALTinterwordstretchfactor\fontdimen3\font minus
		\fontdimen4\font\relax}
	\providecommand{\BIBforeignlanguage}[2]{{%
			\expandafter\ifx\csname l@#1\endcsname\relax
			\typeout{** WARNING: IEEEtran.bst: No hyphenation pattern has been}%
			\typeout{** loaded for the language `#1'. Using the pattern for}%
			\typeout{** the default language instead.}%
			\else
			\language=\csname l@#1\endcsname
			\fi
			#2}}
	\providecommand{\BIBdecl}{\relax}
	\BIBdecl
	
	\bibitem{Nair2007}
	G.~N. Nair, F.~Fagnani, S.~Zampieri, and R.~J. Evans, ``Feedback control under
	data rate constraints: An overview,'' \emph{Proc. IEEE}, vol.~95, pp.
	108--137, 2007.
	
	\bibitem{Ishii2012}
	H.~Ishii and K.~Tsumura, ``{Data rate limitations in feedback control over
		network},'' \emph{IEICE Trans. Fundamentals}, vol. E95-A, pp. 680--690, 2012.
	
	\bibitem{Wakaiki2014IFAC}
	M.~Wakaiki and Y.~Yamamoto, ``{Quantized feedback stabilization of sampled-data
		switched linear systems},'' in \emph{Proc. 19th IFAC WC}, 2014.
	
	\bibitem{Liberzon2014}
	D.~Liberzon, ``Finite data-rate feedback stabilization of switched and hybrid
	linear systems,'' \emph{Automatica}, vol.~50, pp. 409--420, 2014.
	
	\bibitem{Wakaiki2014CDC}
	M.~Wakaiki and Y.~Yamamoto, ``{Output feedback stabilization of switched linear
		systems with limited information},'' in \emph{Proc. 53rd IEEE CDC}, 2014.
	
	\bibitem{Yang2015ACC}
	G.~Yang and D.~Liberzon, ``Stabilizing a switched linear system with
	disturbance by sampled-data quantized feedback,'' in \emph{Proc. ACC'15},
	2015.
	
	\bibitem{WakaikiMTNS2014}
	M.~Wakaiki and Y.~Yamamoto, ``{Quantized output feedback stabilization of
		switched linear systems},'' in \emph{Proc. MTNS'14}, 2014.
	
	\bibitem{Brockett2000}
	R.~W. Brockett and D.~Liberzon, ``Quantized feedback stabilization of linear
	systems,'' \emph{IEEE Trans. Automat. Control}, vol.~45, pp. 1279--1289,
	2000.
	
	\bibitem{Liberzon2003Automatica}
	D.~Liberzon, ``Hybrid feedback stabilization of systems with quantized
	signals,'' \emph{Automatica}, vol.~39, pp. 1543--1554, 2003.
	
	\bibitem{Xiaowu2013}
	M.~Xiaowu and G.~Yang, ``Global input-to-state stabilization with quantized
	feedback for discrete-time piecewise affine systems with time delays,''
	\emph{J. Syst. Sci. Complexity}, vol.~26, pp. 925--939, 2013.
	
	\bibitem{Liberzon2000workshop}
	D.~Liberzon, ``Nonlinear stabilization by hybrid quantized feedback,'' in
	\emph{Proc. HSCC'00}, 2000.
	
	\bibitem{Feng2002}
	G.~Feng, ``{Stability analysis of piecewise discrete-time linear systems},''
	\emph{IEEE Trans. Automat. Control}, vol.~47, pp. 1108--1112, 2002.
	
	\bibitem{Ferrari2002}
	G.~Ferrari-Trecate, F.~A. Cuzzola, D.~Mignone, and M.~Morari, ``Analysis of
	discrete-time piecewise affine and hybrid systems,'' \emph{Automatica},
	vol.~38, pp. 2139--2146, 2002.
	
	\bibitem{Cuzzola2002}
	F.~A. Cuzzola and M.~Morari, ``{An LMI approach for $H_{\infty}$ analysis and
		control of discrete-time piecewise affine systems},'' \emph{Int. J. Control},
	vol.~75, pp. 1293--1301, 2002.
	
	\bibitem{Lazar2008}
	M.~Lazar and W.~P. M.~H. Heemels, ``{Global input-to-state stability and
		stabilization of discrete-time piecewise affine systems},'' \emph{Non},
	vol.~2, pp. 721--734, 2008.
	
	\bibitem{Xu2014}
	J.~Xu and L.~Xie, \emph{Control and Estimation of Piecewise Affine
		Systems}.\hskip 1em plus 0.5em minus 0.4em\relax Woodhead Publishing, 2014.
	
	\bibitem{Bullo2006}
	F.~Bullo and D.~Liberzon, ``{Quantized control via locational optimization},''
	\emph{IEEE Trans. Automat. Control}, vol.~51, pp. 2--13, 2006.
	
	\bibitem{Qiu2011}
	J.~Qiu, G.~Feng, and H.~Gao, ``{Approaches to robust $\mathcal{H}_{\infty}$
		static output feedback control of discrete-time piecewise-affine systems with
		norm-bounded uncertainties},'' \emph{Int. J. Robust and Nonlinear Control},
	vol.~21, pp. 790--814, 2011.
	
	\bibitem{Lin_reachability2011}
	Z.~Lin, M.~Wu, and G.~Yan, ``{Reachability and stabilization of discrete-time
		affine systems with disturbances},'' \emph{Automatica}, vol.~47, pp.
	2720--2727, 2011.
	
	\bibitem{Bemporad2000HSCC}
	A.~Bemporad, F.~D. Torrisi, and M.~Morari, ``{Optimization-based verification
		and stability characterization of piecewise affine and hybrid systems},'' in
	\emph{Proc. HSCC'00}, 2000.
	
	\bibitem{Blanchini2008}
	F.~Blanchini and S.~Miani, \emph{Set-Theory Methods in Control}.\hskip 1em plus
	0.5em minus 0.4em\relax Berlin, Germany: Springer, 2008.
	
	\bibitem{Barmish1979}
	B.~R. Barmish and J.~Sankaran, ``The propagation of parametric uncertainty via
	polytopes,'' \emph{IEEE Trans. Automat. Control}, vol.~24, pp. 346--349,
	1979.
	
	\bibitem{Rubagotti2013}
	M.~Rubagotti, S.~Trimboli, and A.~Bemporad, ``Stability and invariance analysis
	of uncertain discrete-time piecewise affine systems,'' \emph{IEEE Trans.
		Automat. Control}, vol.~58, pp. 2359--2365, 2013.
	
	\bibitem{Ghaoui1997}
	L.~E. Ghaoui, F.~Oustry, and M.~AitRami, ``{A cone complementarity
		linearization algorithm for static output-feedback and related problems},''
	\emph{IEEE Trans. Automat. Control}, vol.~42, pp. 1171--1176, 1997.
	
	\bibitem{Oliveria1997}
	M.~C. de~Oliveria and J.~C. Geromel, ``Numerical comparison of output feedback
	design methods,'' in \emph{Proc. ACC'97}, 1997.
	
\end{thebibliography}
\end{document}